\documentclass[journal,12pt,draftcls,onecolumn]{IEEEtran}

\usepackage[cmex10]{amsmath}
\usepackage{amssymb}
\usepackage{epsfig}
\usepackage{subfigure}
\usepackage{graphicx}

\def\BibTeX{{\rm B\kern-.05em{\sc i\kern-.025em b}\kern-.08em
    T\kern-.1667em\lower.7ex\hbox{E}\kern-.125emX}}

\newtheorem{claim}{Claim}

\newtheorem{theorem}{Theorem}
\newtheorem{definition}{Definition}
\newtheorem{lemma}{Lemma}

\newtheorem{remarks}{Remark}

\newtheorem{example}{Example}

\title{Exact Regeneration Codes for Distributed Storage Repair Using Interference Alignment}

\author{Changho Suh and Kannan Ramchandran \\
Wireless Foundations, Department of EECS \\
University of California at Berkeley \\
Email: \{chsuh, kannanr\}@eecs.berkeley.edu}

\begin{document}

\IEEEaftertitletext{
\begin{abstract}
The high repair cost of $(n,k)$ Maximum Distance Separable (MDS) erasure codes has recently motivated a new class of codes, called {\em Regenerating Codes}, that optimally trade off storage cost for repair bandwidth. On one end of this spectrum of Regenerating Codes are Minimum Storage Regenerating (MSR) codes that can match the minimum storage cost of MDS codes while also significantly reducing repair bandwidth. In this paper, we describe {\em Exact}-MSR codes which allow for \emph{any} failed nodes (\emph{whether they are systematic or parity nodes}) to be regenerated {\em exactly} rather than only functionally or information-equivalently. We show that Exact-MSR codes {\em  come with no loss of optimality} with respect to random-network-coding based MSR codes (matching the cutset-based lower bound on repair bandwidth) for the cases of: $(a)$ $k/n \leq 1/2$; and $(b)$ $k \leq 3$. Our constructive approach is based on {\em interference alignment} techniques, and, unlike the previous class of random-network-coding based approaches, we provide explicit and deterministic coding schemes that require a finite-field size of \emph{at most} $2(n-k)$.

\end{abstract}
\begin{keywords}
Interference Alignment, Minimum Storage Regenerating (MSR) Codes, Repair Bandwidth
\end{keywords}
}

\maketitle

%

\section{Introduction}

In distributed storage systems, maximum distance separable (MDS) erasure codes are well-known coding schemes that can offer maximum reliability for a given storage overhead. For an  $(n,k)$ MDS code for storage, a source file of size $\mathcal{M}$ bits is divided equally into $k$ units (of size $\frac{\mathcal{M}}{k}$ bits each), and these $k$ data units are expanded into $n$ encoded units, and stored at $n$ nodes. The code guarantees that a user or Data Collector (DC) can reconstruct the source file by connecting to any arbitrary $k$ nodes. In other words, any $(n-k)$ node failures can be tolerated with a minimum storage cost of  $\frac{\mathcal{M}}{k}$ at each of $n$ nodes. While MDS codes are optimal in terms of reliability versus storage overhead, they come with a significant maintenance overhead when it comes to repairing failed encoded nodes to restore the MDS system-wide property.
Specifically, consider failure of a single encoded node and the cost needed to restore this node. It can be shown that this repair incurs an aggregate cost of $\mathcal{M}$ bits of information from $k$ nodes. Since each encoded unit contains only $\frac{\mathcal{M}}{k}$ bits of information, this represents a $k$-fold inefficiency with respect to the repair bandwidth.

This challenge has motivated a new class of coding schemes, called \emph{Regenerating Codes} \cite{Dimakis:INFOCOM, Wu:Allerton}, which target the information-theoretic optimal tradeoff between storage cost and repair bandwidth. On one end of this spectrum of Regenerating Codes are Minimum Storage Regenerating (MSR) codes that can match the minimum storage cost of MDS codes while also significantly reducing repair bandwidth. As shown in \cite{Dimakis:INFOCOM, Wu:Allerton}, the fundamental tradeoff between bandwidth and storage depends on the number of nodes that are connected to repair a failed node, simply called the degee $d$ where $k \leq d \leq n-1$. The optimal tradeoff is characterized by
\begin{align}
\label{eq-MSRpoint}
 ( \alpha, \gamma) = \left( \frac{\mathcal{M}}{k}, \frac{\mathcal{M}}{k} \cdot \frac{d}{d-k+1}, \right),
\end{align}
where $\alpha$ and $\gamma$ denote the optimal storage cost and repair bandwidth, respectively for repairing a single failed node, while retaining the MDS-code property for the user.
Note that this code requires the same minimal storage cost (of size $\frac{\mathcal{M}}{k}$) as that of conventional MDS codes, while substantially reducing repair bandwidth by a factor of $\frac{k(d-k+1)}{d}$ (e.g., for $(n,k,d)=(31,6,30)$, there is a $5$x bandwidth reduction). In this paper, without loss of generality, we normalize the repair-bandwidth-per-link $(\frac{\gamma}{d})$ to be 1, making $\mathcal{M}=k(d-k+1)$. One can partition a whole file into smaller chunks so that each has a size of $k(d-k+1)$\footnote{In practice, the order of a file size is of $10^3$ (Kb) $\sim$ $10^9$ (Gb). Hence, it is reasonable to consider this arbitrary size of the chunk.}.
%
%
%



While MSR codes enjoy substantial benefits over MDS codes, they come with some limitations in construction. Specifically, the achievable schemes in \cite{Dimakis:INFOCOM, Wu:Allerton} that meet the optimal tradeoff bound of (\ref{eq-MSRpoint}) restore failed nodes in a \emph{functional} manner only, using a random-network-coding based framework. This means that the replacement nodes maintain the MDS-code property (that any $k$ out of $n$ nodes can allow for the data to be reconstructed) but do not \emph{exactly} replicate the information content of the failed nodes.

Mere functional repair can be limiting. First, in many applications of interest, there is a need to maintain the code in  \emph{systematic} form, i.e., where  the user data in the form of $k$ information units are exactly stored at $k$ nodes and parity information (mixtures of $k$ information units) are stored at the remaining $(n-k)$ nodes. Secondly, under functional repair, additional overhead information needs to be exchanged for \emph{continually} updating repairing-and-decoding rules whenever a failure occurs. This can significantly increase system overhead. A third problem is that the random-network-coding based solution of \cite{Dimakis:INFOCOM} can require a huge finite-field size, which can significantly increase the computational complexity of encoding-and-decoding\footnote{In \cite{Dimakis:INFOCOM}, Dimakis-Godfrey-Wu-Wainwright-Ramchandran translated the regenerating-codes problem into a multicast communication problem where random-network-coding-based schemes require a huge field size especially for large networks. In storage problems, the field size issue is further aggravated by the need to support a dynamically expanding network size due to the need for continual repair.}. Lastly, functional repair is undesirable in storage security applications in the face of eavesdroppers. In this case, information leakage occurs continually due to the dynamics of repairing-and-decoding rules that can be potentially observed by eavesdroppers \cite{Sameer:ISIT2010}.

These drawbacks motivate the need for \emph{exact} repair of failed nodes. This leads to the following question: is there a price for attaining the optimal tradeoff of (\ref{eq-MSRpoint}) with the extra constraint of exact repair? The work in \cite{KumarRamchandran_MSR} sheds some light on this question: specifically, it was shown that under \emph{scalar linear} codes\footnote{In scalar linear codes, symbols are not allowed to be split into arbitrarily small sub-symbols as with vector linear codes. This is equivalent to having large block-lengths in the classical setting. Under non-linear and vector linear codes, whether or not the optimal tradeoff can be achieved for this regime remains open.}, when  $\frac{k}{n} > \frac{1}{2} + \frac{2}{n}$, there \emph{is} a price for exact repair. For large $n$, this case boils down to $\frac{k}{n} > \frac{1}{2}$, i.e., redundancy less than two. Now what about for $\frac{k}{n} \leq \frac{1}{2}$? This paper resolves this open problem and \emph{shows that it is indeed possible to attain the optimal tradeoff of (\ref{eq-MSRpoint}) for the case of $\frac{k}{n} \leq \frac{1}{2}$ (and $d \geq 2k - 1$), while also guaranteeing exact repair}. Furthermore, we show that for the special case of $k \leq 3$, there is no price for exact repair, regardless of the value of $n$. The interesting special case in this class is the $(5,3)$ Exact-MSR code\footnote{Independently, Cullina-Dimakis-Ho in \cite{Cullina_MSR} found $(5,3)$ E-MSR codes defined over ${\sf GF}(3)$, based on a search algorithm.}, which is not covered by the first case of $\frac{k}{n} \leq \frac{1}{2}$.

Our achievable scheme builds on the concept of \emph{interference alignment}, which was introduced in the context of wireless communication networks \cite{Mohammad,Jafar:IC}.
The idea of interference alignment is to align multiple interference signals in a signal subspace whose dimension is smaller than the number of interferers. Specifically, consider the following setup where a decoder has to decode one desired signal which is linearly interfered with by two separate undesired signals. How many linear equations (relating to the number of channel uses) does the decoder need to recover its desired input signal? As the aggregate signal dimension spanned by desired and undesired signals is at most three, the decoder can naively recover its signal of interest with access to three linearly independent equations in the three unknown signals. However, as the decoder is interested in only one of the three signals, it can decode its desired unknown signal even if it has access to only two equations, provided the two undesired signals are judiciously aligned in a 1-dimensional subspace. See \cite{Mohammad,Jafar:IC,Suh:Allerton} for details.

We will show in the sequel how this concept relates intimately to our repair problem. At a high level, the connection comes from our repair problem involving recovery of a subset (related to the subspace spanned by a failed node) of the overall aggregate signal space (related to the entire user data dimension). There are, however, significant differences some beneficial and some detrimental. On the positive side, while in the wireless problem, the equations are provided by \emph{nature} (in the form of channel gain coefficients), in our repair problem, the coefficients of the equations are \emph{man-made} choices, representing a part of the overall design space. On the flip side, however, the MDS requirement of our storage code and the multiple failure configurations that need to be simultaneously addressed with a single code design generate multiple interference alignment constraints that need to be simultaneously satisfied. This is particularly acute for a large value of $k$, as the number of possible failure configurations increases with $n$ (which increases with $k$). Finally, another difference comes from the finite-field constraint of our repair problem.

We propose a \emph{common-eigenvector} based conceptual framework (explained in Section~\ref{sec-BasisFramework}) that covers all possible failure configurations. Based on this framework, we develop an interference alignment design technique for exact repair. We also propose another interference alignment scheme for a $(5,3)$ code\footnote{The finite-field nature of the problem makes this challenging.}, which in turn shows the optimality of the cutset bound (\ref{eq-MSRpoint}) for the case $k \leq 3$. As in \cite{KumarRamchandran_MSR}, our coding schemes are deterministic and require a field size of \emph{at most} $2(n-k)$. This is in stark contrast to the random-network-coding based solutions \cite{Dimakis:INFOCOM}.

\begin{figure}[t]
\begin{center}
{\epsfig{figure=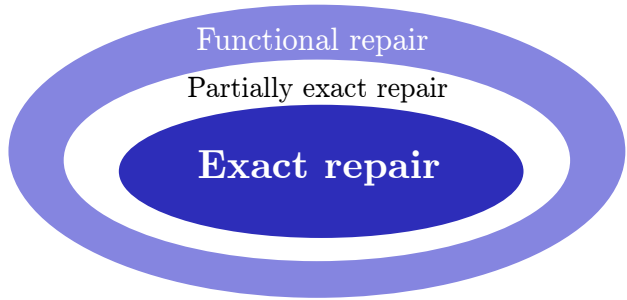, angle=0, width=0.5\textwidth}}
\end{center}
\caption{Repair models for distributed storage systems. In exact repair, the failed nodes are exactly regenerated, thus restoring lost encoded fragments with their exact replicas. In functional repair, the requirement is relaxed: the newly generated node can contain different data from that of the failed node as long as the repaired system maintains the MDS-code property. In partially exact repair, only systematic nodes are repaired exactly, while parity nodes are repaired only functionally.}\label{fig:VanDiagram}
\end{figure}

\section{Connection to Related Work}
As stated earlier, Regenerating Codes, which cover an entire spectrum of optimal tradeoffs between repair bandwidth and storage cost, were introduced in \cite{Dimakis:INFOCOM, Wu:Allerton}. As discussed, MSR codes occupy one end of this spectrum corresponding to minimum storage. At the other end of the spectrum live Minimum Bandwidth Regenerating (MBR) codes corresponding to minimum repair bandwidth. The optimal tradeoffs described in \cite{Dimakis:INFOCOM, Wu:Allerton} are based on random-network-coding based approaches,  which guarantee only functional repair.

The topic of exact repair codes has received attention in the recent literature \cite{Wu:ISIT, KumarRamchandran_MBR, KumarRamchandran_MSR, Cullina_MSR, Wu:PartialExact}. Wu and Dimakis in \cite{Wu:ISIT} showed that the MSR point (\ref{eq-MSRpoint}) can be attained for the cases of: $k=2$ and $k=n-1$. Rashmi-Shah-Kumar-Ramchandran in  \cite{KumarRamchandran_MBR} showed that for $d=n-1$, the optimal MBR point can be achieved with a deterministic scheme requiring a small finite-field size and zero repair-coding-cost. Subsequently, Shah-Rashmi-Kumar-Ramchandran in \cite{KumarRamchandran_MSR} developed partially exact codes for the MSR point corresponding to $ \frac{k}{n} \leq \frac{1}{2} + \frac{2}{n}$, where exact repair is limited to the systematic component of the code. See Fig.~\ref{fig:VanDiagram}. Finding the fundamental limits under exact repair of \emph{all} nodes (including parity) remained an open problem. A key contribution of this paper is to resolve this open problem by showing that E-MSR codes come with no extra cost over the optimal tradeoff of (\ref{eq-MSRpoint}) for the case of $\frac{k}{n} \leq \frac{1}{2}$ (and $d  \geq 2k - 1$). For the most general case, finding the fundamental limits under exact repair constraints for all values of $(n,k,d)$ remains an open problem.

The constructive framework proposed in \cite{KumarRamchandran_MSR} forms the inspiration for our proposed solution in this paper. Indeed, we show that the code introduced in [4] for exact repair of only the systematic nodes can also be used to repair the non-systematic (parity) node failures exactly \emph{provided repair construction schemes are appropriately designed}. This design for ensuring exact repair of \emph{all} nodes is challenging and had remained an open problem: resolving this for the case of $\frac{k}{n} \leq \frac{1}{2}$ (and $d \geq 2k-1$) is a key contribution of this work. Another contribution of our work is the systematic development of a generalized family of code structures (of which the code structure of \cite{KumarRamchandran_MSR} is a special case), together with the associated optimal repair construction schemes. This generalized family of codes provides conceptual insights into the structure of solutions for the exact repair problem, while also opening up a much larger constructive design space of solutions.



\section{Interference Alignment for Distributed Storage Repair}
\label{sec-Notations}
Linear network coding \cite{Koetter:it,ahlswede:it} (that allows multiple messages to be linearly combined at network nodes) has been established recently as a useful tool for addressing interference issues even in wireline networks where all the communication links are orthogonal and non-interfering. This attribute was first observed in \cite{Wu:ISIT}, where it was shown that interference alignment could be exploited for storage networks, specifically for minimum storage regenerating (MSR) codes having small $k$ ($k=2$). However, generalizing interference alignment to large values of $k$ (even $k=3$) proves to be challenging, as we describe in the sequel. In order to appreciate this better, let us first review the scheme of \cite{Wu:ISIT} that was applied to the exact repair problem. We will then address the difficulty of extending interference alignment for larger systems and describe how to address this in Section \ref{sec-BasisFramework}.

\subsection{Review of $(4,2)$ E-MSR Codes \cite{Wu:ISIT}}
Fig. \ref{fig:42example} illustrates an interference alignment scheme for a $(4,2)$ MDS code defined over ${\sf GF}(5)$. First one can easily check the MDS property of the code, i.e., all the source files can be reconstructed from any $k(=2)$ nodes out of $n(=4)$ nodes. Let us see how failed node 1 (storing $(a_1, a_2)$) can be exactly repaired. We assume that the degree $d$ (the number of storage nodes connected to repair a failed node) is $3$, and a source file size  $\mathcal{M}$ is $4$. The cutset bound (\ref{eq-MSRpoint}) then gives the fundamental limits of: storage cost $\alpha = 2$; and repair-bandwidth-per-link $\frac{\gamma}{d}=1$.

The example illustrated in Fig. \ref{fig:42example} shows that the parameter set described above is achievable using interference alignment. Here is a summary of the scheme. First notice that since the bandwidth-per-link is 1, two symbols in each storage node are projected into a scalar variable with projection weights. Choosing appropriate weights, we get the equations as shown in Fig. \ref{fig:42example}: $(b_1 + b_2)$; $a_1 + 2a_2 + (b_1 + b_2)$; $2 a_1 + a_2 + (b_1 + b_2)$. Observe that the undesired signals $(b_1,b_2)$ (interference) are aligned onto an 1-dimensional linear subspace, thereby achieving \emph{interference alignment}. Therefore, we can successfully decode $(a_1,a_2)$ with three equations although there are four unknowns.
\begin{figure}[t]
\begin{center}
{\epsfig{figure=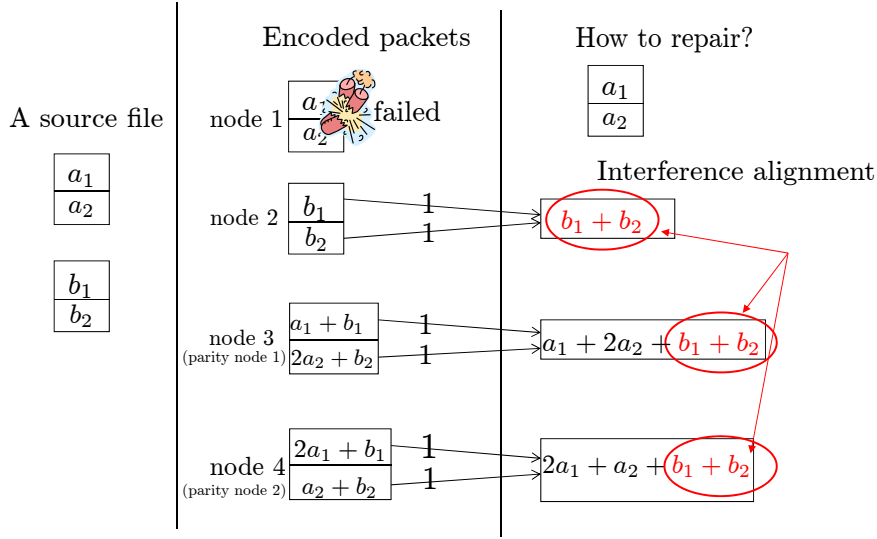, angle=0, width=0.7\textwidth}}
\end{center}
\caption{Interference alignment for a $(4,2)$ E-MSR code defined over ${\sf GF}(5)$ \cite{Wu:ISIT}. Choosing appropriate projection weights, we can align interference space of $(b_1, b_2)$ into one-dimensional linear space spanned by $[1,\; 1]^t$. As a result, we can successfully decode 2 desired unknowns $(a_1,a_2)$ from 3 equations containing 4 unknowns $(a_1,a_2,b_1,b_2)$.} \label{fig:42example}

\end{figure}
Similarly, we can repair $(b_1,b_2)$ when it has failed.

\begin{figure}[t]
\begin{center}
{\epsfig{figure=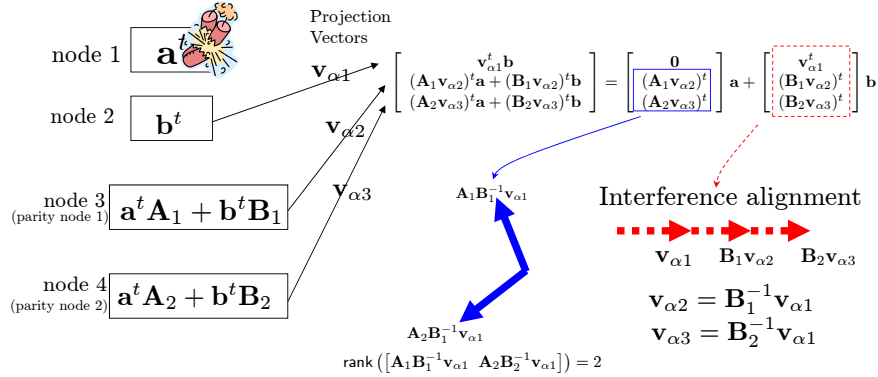, angle=0, width=0.7\textwidth}}
\end{center}
\caption{Geometric interpretation of interference alignment. The blue solid-line and red dashed-line vectors indicate linear subspaces with respect to ``$\mathbf{a}$'' and ``$\mathbf{b}$'', respectively. The choice of $\mathbf{v}_{\alpha 2} = \mathbf{B}_1^{-1} \mathbf{v}_{\alpha 1}$ and $\mathbf{v}_{\alpha 3} = \mathbf{B}_2^{-1} \mathbf{v}_{\alpha 1}$ enables interference alignment. For the specific example of Fig. \ref{fig:42example}, the corresponding encoding matrices are
$\mathbf{A}_1 = \left[1, 0; 0, 2 \right]$, $\mathbf{B}_1 = \left[1, 0; 0,1 \right]$.
$\mathbf{A}_2 = \left[ 2,  0; 0, 1 \right]$, $\mathbf{B}_2 = \left[ 1,  0; 0,  1 \right]$.} \label{fig:GeometricView}
\end{figure}

\subsection{Matrix Notation}
We introduce matrix notation that provides geometric interpretation of interference alignment and is useful for generalization. Let $\mathbf{a}= (a_1,a_2)^t$ and $\mathbf{b}= (b_1,b_2)^t$  be 2-dimensional information-unit vectors, where $(\cdot)^t$ indicates a transpose. Let $\mathbf{A}_i$ and $\mathbf{B}_i$ be $2$-by-$2$ encoding matrices for parity node $i$ ($i=1,2$), which contain encoding coefficients for the linear combination of the components of ``$\mathbf{a}$'' and ``$\mathbf{b}$''. For example, parity node 1 stores information in the form of $\mathbf{a}^t \mathbf{A}_1  + \mathbf{b}^t \mathbf{B}_1 $, as shown in Fig.~\ref{fig:GeometricView}. The encoding matrices for systematic nodes are not explicitly defined since those are trivially inferred. Finally we define 2-dimensional projection vectors $\mathbf{v}_{\alpha i}$'s ($i=1,2,3$).

Let us consider exact repair of systematic node 1. By connecting to three nodes, we get: $\mathbf{b}^{t} \mathbf{v}_{\alpha 1}$; $\mathbf{a}^t (\mathbf{A}_1 \mathbf{v}_{\alpha 2}) +\mathbf{b}^t (\mathbf{B}_1  \mathbf{v}_{\alpha 2})$;
$ \mathbf{a}^t ( \mathbf{A}_2 \mathbf{v}_{\alpha 3} ) + \mathbf{b}^t ( \mathbf{B}_2  \mathbf{v}_{\alpha 3} )$. Recall the goal, which is to decode 2 desired unknowns out of 3 equations including 4 unknowns. To achieve this goal, we need:
\begin{align}
{\sf rank} \left( \left[
    \begin{array}{c}
      (\mathbf{A}_1 \mathbf{v}_{\alpha 2})^t \\
      ( \mathbf{A}_2 \mathbf{v}_{\alpha 3} )^t \\
    \end{array}
  \right] \right) = 2; \;\; {\sf rank} \left( \left[
    \begin{array}{c}
      \mathbf{v}_{\alpha 1}^t \\
      (\mathbf{B}_1  \mathbf{v}_{\alpha 2})^t \\
      ( \mathbf{B}_2  \mathbf{v}_{\alpha 3} )^t \\
    \end{array}
  \right] \right) = 1.
\end{align}
The second condition can be met by setting $\mathbf{v}_{\alpha 2} = \mathbf{B}_1^{-1} \mathbf{v}_{\alpha 1}$ and $\mathbf{v}_{\alpha 3} = \mathbf{B}_2^{-1} \mathbf{v}_{\alpha 1}$. This choice forces the interference space to be collapsed into a one-dimensional linear subspace, thereby achieving interference alignment. With this setting, the first condition now becomes
\begin{align}
\label{eq_42_1}
\mathsf{rank} \left( \left[ \mathbf{A}_1 \mathbf{B}_1^{-1} \mathbf{v}_{\alpha 1} \;\;
\mathbf{A}_2  \mathbf{B}_2^{-1} \mathbf{v}_{\alpha 1} \right] \right) = 2.
\end{align}
It can be easily verified that the choice of $\mathbf{A}_i$'s and $\mathbf{B}_i$'s given in Figs. \ref{fig:42example} and \ref{fig:GeometricView} guarantees the above condition.
When the node 2 fails, we get a similar condition:
\begin{align}
\label{eq_42_2}
\mathsf{rank} \left( \left[ \mathbf{B}_1 \mathbf{A}_1^{-1} \mathbf{v}_{\beta 1} \;\;
\mathbf{B}_2  \mathbf{A}_2^{-1} \mathbf{v}_{\beta 1} \right] \right) = 2,
\end{align}
where $\mathbf{v}_{\beta i}$'s denote projection vectors for node 2 repair. This condition also holds under the given choice of encoding matrices.

\subsection{Connection with Interference Channels in Communication Problems}
Observe the three equations shown in Fig. \ref{fig:GeometricView}:
\begin{align*}
\underbrace{\left[
    \begin{array}{c}
      \mathbf{0} \\
      (\mathbf{A}_1 \mathbf{v}_{\alpha 2})^t \\
      ( \mathbf{A}_2 \mathbf{v}_{\alpha 3} )^t \\
    \end{array}
  \right] \mathbf{a}}_{desired\;signals} + \underbrace{\left[
    \begin{array}{cc}
      \mathbf{v}_{\alpha 1}^t \\
      (\mathbf{B}_1  \mathbf{v}_{\alpha 2})^t \\
      ( \mathbf{B}_2  \mathbf{v}_{\alpha 3} )^t \\
    \end{array}
  \right] \mathbf{b}}_{interference}.
\end{align*}
Separating into two parts, we can view this problem as a wireless communication problem, wherein a subset of the information is desired to be decoded in the presence of interference. Note that for each term (e.g., $\mathbf{A}_1 \mathbf{v}_{\alpha 2}$), the matrix $\mathbf{A}_1$ and vector $\mathbf{v}_{\alpha 2}$ correspond to channel matrix and transmission vector in wireless communication problems, respectively.

There are, however, significant differences. In the wireless communication problem, the channel matrices are provided \emph{by nature} and therefore not controllable. The transmission strategy alone (vector variables) can be controlled for achieving interference alignment. On the other hand, in our storage repair problems, both matrices and vectors are controllable, i.e., projection vectors and encoding matrices can be arbitrarily designed, resulting in more flexibility. However, our storage repair problem comes with unparalleled challenges due to the MDS requirement and the  multiple failure configurations. These induce multiple interference alignment constraints that need to be simultaneously satisfied.
What makes this difficult is that the encoding matrices, once designed, must be the same for \emph{all} repair configurations. This is particularly acute for large values of $k$ (even $k=3$), as the number of possible failure configurations increases with $n$ (which increases with $k$).


\section{A Proposed Framework for Exact-MSR Codes}
\label{sec-BasisFramework}


We propose a \emph{common-eigenvector} based conceptual framework to address the exact repair problem. This framework draws its inspiration from the work in \cite{KumarRamchandran_MSR} which guarantees the exact repair of systematic nodes, while satisfying the MDS code property, but which does not provide exact repair of failed parity nodes. In providing a solution for the exact repair of \emph{all} nodes, we propose here a generalized family of codes (of which the code in \cite{KumarRamchandran_MSR} is a special case). This both provides insights into the structure of codes for exact repair of all nodes, as well as opens up a much larger design space for constructive solutions. Specifically, we propose a \emph{common-eigenvector} based approach building on a certain \emph{elementary matrix} property \cite{Householder, Dubrulle} for the generalized code construction. Moreover, as in \cite{KumarRamchandran_MSR}, our proposed coding schemes are deterministic and constructive, requiring a symbol alphabet-size of at most $(2n-2k)$.

Our framework consists of four components:
(1) developing a family of codes\footnote{Interestingly, the structure of the code in \cite{KumarRamchandran_MSR} turns out to work for the exact repair of both systematic and parity nodes provided appropriate repair schemes are developed.} for exact repair of systematic codes based on the common-eigenvector concept; (2) drawing a \emph{dual} relationship between the systematic and parity node repair; (3) guaranteeing the MDS property of the code; (4) constructing codes with finite-field alphabets. 
The framework covers the case of $n \geq 2k$ (and $d \geq 2k-1$). It turns out that the $(2k,k, 2k-1)$ code case contains the key design ingredients and the case of $n \geq 2k$ can be derived from this (see Section \ref{sec-generalization}). Hence, we first focus on the simplest example: $(6,3,5)$ E-MSR codes. Later in Section \ref{sec-generalization}, we will generalize this to arbitrary $(n,k,d)$ codes in the class.


\begin{figure}[t]
\begin{center}
{\epsfig{figure=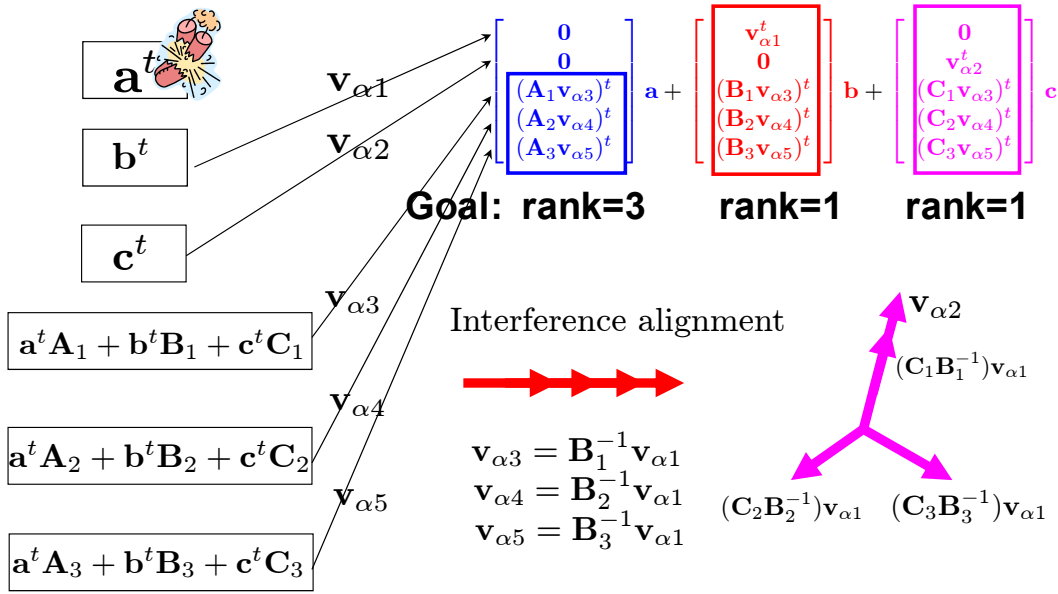, angle=0, width=0.85\textwidth}}
\end{center}
\caption{Difficulty of achieving interference alignment simultaneously.}\label{fig:63_EMSR_Challenge}
\end{figure}

\subsection{Code Structure for Systematic Node Repair}

For $k \geq 3$ (more-than-two interfering information units), achieving interference alignment for exact repair turns out to be significantly more complex than the $k=2$ case. Fig.~\ref{fig:63_EMSR_Challenge} illustrates this difficulty through the example of repairing node 1 for a $(6,3,5)$ code.
By the optimal tradeoff (\ref{eq-MSRpoint}), the choice of $\mathcal{M}=9$ gives $\alpha=3$ and $\frac{\gamma}{d}=1$. Let $\mathbf{a}=(a_1,a_2,a_3)^t$, $\mathbf{b}=(b_1,b_2,b_3)^t$ and $\mathbf{c}=(c_1,c_2,c_3)^t$. We define 3-by-3 encoding matrices of $\mathbf{A}_i$, $\mathbf{B}_i$ and  $\mathbf{C}_i$ (for $i=1,2,3$); and 3-dimensional projection vectors $\mathbf{v}_{\alpha i}$'s.

Consider the 5 $(=d)$ equations downloaded from the nodes:
\begin{align*}
  \left[
    \begin{array}{c}
      \mathbf{0} \\
      \mathbf{0}  \\
      (\mathbf{A}_1 \mathbf{v}_{\alpha 3})^t \\
       ( \mathbf{A}_2 \mathbf{v}_{\alpha 4} )^t  \\
      ( \mathbf{A}_3 \mathbf{v}_{\alpha 5} )^t
    \end{array}
  \right] \mathbf{a} +
 \left[
    \begin{array}{c}
      \mathbf{v}_{\alpha 1}^t  \\
       \mathbf{0} \\
        (\mathbf{B}_1  \mathbf{v}_{\alpha 3})^t \\
       ( \mathbf{B}_2  \mathbf{v}_{\alpha 4} )^t \\
        ( \mathbf{B}_3  \mathbf{v}_{\alpha 5} )^t \\
    \end{array}
  \right] \mathbf{b} + \left[
    \begin{array}{c}
       \mathbf{0} \\
        \mathbf{v}_{\alpha 2}^t \\
     (\mathbf{C}_1  \mathbf{v}_{\alpha 3})^t \\
        (\mathbf{C}_2  \mathbf{v}_{\alpha 4})^t \\
       (\mathbf{C}_3  \mathbf{v}_{\alpha 5})^t\\
    \end{array}
  \right] \mathbf{c}.
  \end{align*}
In order to successfully recover the desired signal components of ``$\mathbf{a}$'', the matrices associated with $\mathbf{b}$ and $\mathbf{c}$ should have rank 1, respectively, while the matrix associated with $\mathbf{a}$ should have full rank of 3. In accordance with the $(4,2)$ code example in Fig.~\ref{fig:GeometricView}, if one were to set $\mathbf{v}_{\alpha 3} = \mathbf{B}_1^{-1} \mathbf{v}_{\alpha 1}$, $\mathbf{v}_{\alpha 4} = \mathbf{B}_2^{-1} \mathbf{v}_{\alpha 2}$ and $\mathbf{v}_{\alpha 5} = \mathbf{B}_3^{-1} \mathbf{v}_{\alpha 1}$, then it is possible to achieve interference alignment with respect to $\mathbf{b}$. However, this choice also specifies the interference space of $\mathbf{c}$. If the $\mathbf{B}_i$'s and $\mathbf{C}_i$'s are not designed judiciously, interference alignment is not guaranteed for $\mathbf{c}$. Hence, it is not evident how to achieve interference alignment \emph{at the same time}.

\begin{figure}[t]
\begin{center}
{\epsfig{figure=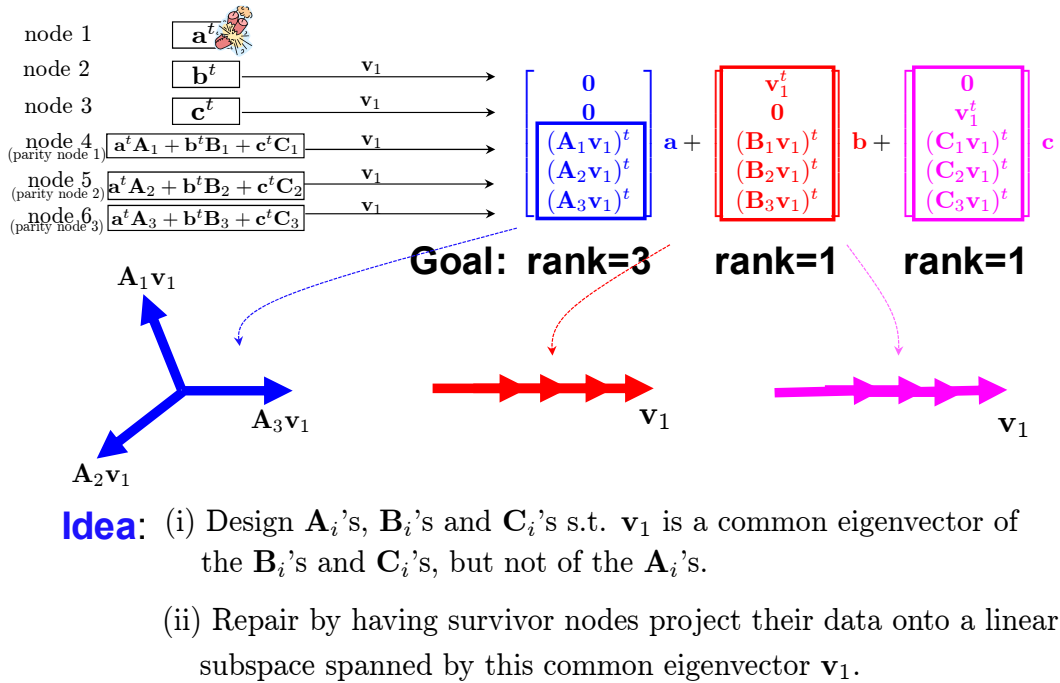, angle=0, width=0.85\textwidth}}
\end{center}
\caption{Illustration of exact repair of systematic node 1 for $(6,3,5)$ E-MSR codes. The idea consists of two parts: (i) designing $(\mathbf{A}_i,\mathbf{B}_i, \mathbf{C}_i)$'s such that $\mathbf{v}_1$ is a common eigenvector of the $\mathbf{B}_i$'s and $\mathbf{C}_i$'s, but not of $\mathbf{A}_i$'s; (ii) repairing by having survivor nodes project their data onto a linear subspace spanned by this common eigenvector $\mathbf{v}_1$.}\label{fig:63_EMSR}
\end{figure}

In order to address the challenge of simultaneous interference alignment, we invoke a \emph{common eigenvector} concept. The idea consists of two parts: (i) designing the $(\mathbf{A}_i,\mathbf{B}_i, \mathbf{C}_i)$'s such that $\mathbf{v}_1$ is a common eigenvector of the $\mathbf{B}_i$'s and $\mathbf{C}_i$'s, but not of $\mathbf{A}_i$'s\footnote{Of course, five additional constraints also need to be satisfied for the other five failure configurations for this $(6,3,5)$ code example.}; (ii) repairing by having survivor nodes \emph{project} their data onto a linear subspace spanned by this common eigenvector $\mathbf{v}_1$.
We can then achieve interference alignment for $\mathbf{b}$ and $\mathbf{c}$ \emph{at the same time}, by setting $\mathbf{v}_{\alpha i} = \mathbf{v}_1, \forall i$. As long as $[\mathbf{A}_1 \mathbf{v}_1, \mathbf{A}_2 \mathbf{v}_1, \mathbf{A}_3 \mathbf{v}_1 ]$ is invertible, we can also guarantee the decodability of $\mathbf{a}$.
See Fig.~\ref{fig:63_EMSR}.

The challenge is now to design encoding matrices to guarantee the existence of a common eigenvector while also satisfying the decodability of desired signals. The difficulty comes from the fact that in our $(6,3,5)$ code example, these constraints need to be satisfied for \emph{all} six possible failure configurations. The structure of \emph{elementary matrices} \cite{Householder, Dubrulle} (generalized matrices of Householder and Gauss matrices) gives insights into this. To see this, consider a 3-by-3 elementary matrix $\mathbf{A}$:
\begin{align}
\mathbf{A} =  \mathbf{u} \mathbf{v}^t + \alpha \mathbf{I},
\end{align}
where $\mathbf{u}$ and $\mathbf{v}$ are 3-dimensional vectors. Here is an observation that motivates our proposed structure: the dimension of the null space of $\mathbf{v}$ is 2 and the null vector $\mathbf{v}^\bot$ is an eigenvector of $\mathbf{A}$, i.e., $\mathbf{A} \mathbf{v}^{\bot} = \alpha \mathbf{v}^{\bot}$.
This motivates the following structure:
\begin{align}
\begin{split}
\label{eq:63_EncodingMatrices}
\mathbf{A}_1 &= \mathbf{u}_1 \mathbf{v}_1^t  + \alpha_1 \mathbf{I}; \; \mathbf{B}_1 = \mathbf{u}_1 \mathbf{v}_2^t  + \beta_1 \mathbf{I}; \;
\mathbf{C}_1 = \mathbf{u}_1 \mathbf{v}_3^t  + \gamma_1 \mathbf{I} \\
\mathbf{A}_2 &= \mathbf{u}_2 \mathbf{v}_1^t  + \alpha_2 \mathbf{I}; \; \mathbf{B}_2 = \mathbf{u}_2 \mathbf{v}_2^t + \beta_2 \mathbf{I}; \;
\mathbf{C}_2 = \mathbf{u}_2 \mathbf{v}_3^t+ \gamma_2 \mathbf{I} \\
\mathbf{A}_3 &= \mathbf{u}_3 \mathbf{v}_1^t  + \alpha_3 \mathbf{I}; \;
\mathbf{B}_3 = \mathbf{u}_3 \mathbf{v}_2^t + \beta_3 \mathbf{I}; \;
\mathbf{C}_3 =\mathbf{u}_3 \mathbf{v}_3^t + \gamma_3 \mathbf{I},
\end{split}
\end{align}
where $\mathbf{v}_i$'s are 3-dimensional linearly independent vectors and so are $\mathbf{u}_i$'s. The values of the $\alpha_i$'s, $\beta_i$'s and $\gamma_i$'s can be arbitrary non-zero values. First consider the simple case where the $\mathbf{v}_i$'s are \emph{orthonormal}. This is for conceptual simplicity. Later we will generalize to the case where the $\mathbf{v}_i$'s need not be orthogonal but only linearly independent: namely, \emph{bi-orthogonal} case. For the orthogonal case, we see that for $i=1,2,3,$
\begin{align}
\begin{split}
\mathbf{A}_i \mathbf{v}_1 &=  \alpha_i \mathbf{v}_1 + \mathbf{u}_i,\\
\mathbf{B}_i \mathbf{v}_1 &= \beta_i \mathbf{v}_1,\\
\mathbf{C}_i \mathbf{v}_1 &= \gamma_i \mathbf{v}_1.
\end{split}
\end{align}
Importantly, notice that $\mathbf{v}_1$ is a common eigenvector of the $\mathbf{B}_i$'s and $\mathbf{C}_i$'s, while simultaneously ensuring that the vectors of $\mathbf{A}_i \mathbf{v}_1$ are linearly independent. Hence, setting $\mathbf{v}_{\alpha i} = \mathbf{v}_1$ for all $i$, it is possible to achieve simultaneous interference alignment while also guaranteeing the decodability of the desired signals. See Fig.~\ref{fig:63_EMSR}. On the other hand, this structure also guarantees exact repair for $\mathbf{b}$ and $\mathbf{c}$. We use $\mathbf{v}_2$ for exact repair of $\mathbf{b}$. It is a common eigenvector of the $\mathbf{C}_i$'s and $\mathbf{A}_i$'s, while ensuring $[\mathbf{B}_1 \mathbf{v}_2, \mathbf{B}_2 \mathbf{v}_2, \mathbf{B}_3 \mathbf{v}_2]$ invertible. Similarly, $\mathbf{v}_3$ is used for $\mathbf{c}$.

We will see that a \emph{dual basis} property gives insights into the general \emph{bi-orthogonal} case where $\{ \mathbf{v} \}:= (\mathbf{v}_1, \mathbf{v}_2, \mathbf{v}_3)$ is not orthogonal but linearly independent. In this case, defining a dual basis $\{ \mathbf{v}' \}:= (\mathbf{v}_1',\mathbf{v}_2', \mathbf{v}_3')$ gives the solution:
\begin{align*}
\left[
  \begin{array}{c}
    \mathbf{v}_1'^t \\
    \hline
    \mathbf{v}_2'^t  \\
    \hline
    \mathbf{v}_3'^t  \\
  \end{array}
\right]: =
\left[
  \begin{array}{c|c|c}
    \mathbf{v}_1 & \mathbf{v}_2 & \mathbf{v}_3 \\
  \end{array}
\right]^{-1}.
\end{align*}
The definition gives the following property: $\mathbf{v}_i'^t \mathbf{v}_j = \delta( i-j), \forall i,j.
$ Using this property, one can see that $\mathbf{v}_1'$ is a common eigenvector of the $\mathbf{B}_i$'s and $\mathbf{C}_i$'s:
 \begin{align}
\begin{split}
\mathbf{A}_i \mathbf{v}_1' &=  \alpha_i \mathbf{v}_1' + \mathbf{u}_i,\\
\mathbf{B}_i \mathbf{v}_1' &= \beta_i \mathbf{v}_1',\\
\mathbf{C}_i \mathbf{v}_1' &= \gamma_i \mathbf{v}_1'.
\end{split}
\end{align}
So it can be used as a projection vector for exact repair of $\mathbf{a}$.
Similarly, we can use $\mathbf{v}_2'$ and $\mathbf{v}_3'$ for exact repair of $\mathbf{b}$ and $\mathbf{c}$, respectively.

\subsection{Dual Relationship between Systematic and Parity Node Repair}
We have seen so far how to ensure exact repair of the systematic nodes. We have known that if $\{ \mathbf{v} \}$ is linearly independent and so $\{ \mathbf{u} \}$ is, then using the code-structure of (\ref{eq:63_EncodingMatrices}) together with \emph{projection direction} enables repair, for arbitrary values of $( \alpha_i, \beta_i, \gamma_i)$'s. A natural question is now: will this code structure also guarantee exact repair of parity nodes? It turns out that for exact repair of all nodes, we need a special relationship between $\{ \mathbf{v} \}$ and $\{ \mathbf{u} \} $ through the correct choice of the $( \alpha_i, \beta_i, \gamma_i)$'s.

We will show that parity nodes can be repaired by drawing a \emph{dual} relationship with systematic nodes.
The procedure has two steps. The first is to remap parity nodes with
$\mathbf{a}'$, $\mathbf{b}'$, and $\mathbf{c}'$, respectively:
\begin{align*}
\left[
    \begin{array}{ccc}
      \mathbf{a}'\\
      \mathbf{b}'\\
      \mathbf{c}'\\
    \end{array}
  \right]:=
\left[
    \begin{array}{ccc}
      \mathbf{A}_1^t & \mathbf{B}_1^t & \mathbf{C}_1^t\\
      \mathbf{A}_2^t & \mathbf{B}_2^t & \mathbf{C}_2^t\\
      \mathbf{A}_3^t & \mathbf{B}_3^t & \mathbf{C}_3^t\\
    \end{array}
  \right] \left[
    \begin{array}{ccc}
      \mathbf{a}\\
      \mathbf{b}\\
      \mathbf{c}\\
    \end{array}
  \right].
\end{align*}
Systematic nodes can then be rewritten in terms of the prime notations:
\begin{align}
\begin{split}
\mathbf{a}^t &= \mathbf{a}'^t \mathbf{A}_1' + \mathbf{b}'^t \mathbf{B}_1' +
\mathbf{c}'^t \mathbf{C}_1', \\
\mathbf{b}^t &= \mathbf{a}'^t \mathbf{A}_2' + \mathbf{b}'^t \mathbf{B}_2' +
\mathbf{c}'^t \mathbf{C}_2', \\
\mathbf{c}^t &= \mathbf{a}'^t \mathbf{A}_3' + \mathbf{b}'^t \mathbf{B}_3' +
\mathbf{c}'^t \mathbf{C}_3',
\end{split}
\end{align}
where the newly mapped encoding matrices $(\mathbf{A}_i', \mathbf{B}_i', \mathbf{C}_i)$'s are defined as:
\begin{align}
\label{eq-63_mapping}
\left[
   \begin{array}{ccc}
     \mathbf{A}_1'  & \mathbf{A}_2' & \mathbf{A}_3' \\
     \mathbf{B}_1'  & \mathbf{B}_2' & \mathbf{B}_3' \\
\mathbf{C}_1' & \mathbf{C}_2' & \mathbf{C}_3' \\
   \end{array}
 \right]: = \left[
   \begin{array}{ccc}
     \mathbf{A}_1  & \mathbf{A}_2 & \mathbf{A}_3 \\
     \mathbf{B}_1  & \mathbf{B}_2 & \mathbf{B}_3 \\
\mathbf{C}_1 & \mathbf{C}_2 & \mathbf{C}_3 \\
   \end{array}
 \right]^{-1}.
\end{align}
With this remapping, one can dualize the relationship between systematic and parity node repair. Specifically, if all of the $\mathbf{A}_i'$'s, $\mathbf{B}_i'$'s, and $\mathbf{C}_i'$'s are \emph{elementary matrices} and form a similar code-structure as in (\ref{eq:63_EncodingMatrices}), exact repair of the parity nodes becomes transparent.

The challenge is now how to guarantee the dual structure. In Lemma~\ref{lemma:TensorElementaryMatrix}, we show that a special relationship between $\{ \mathbf{u} \}$ and $ \{ \mathbf{v} \}$ through ($\alpha_i$, $\beta_i$, $\gamma_i$)'s can guarantee this dual relationship of (\ref{eq-63_PrimeEM}).


\begin{figure}[t]
\begin{center}
{\epsfig{figure=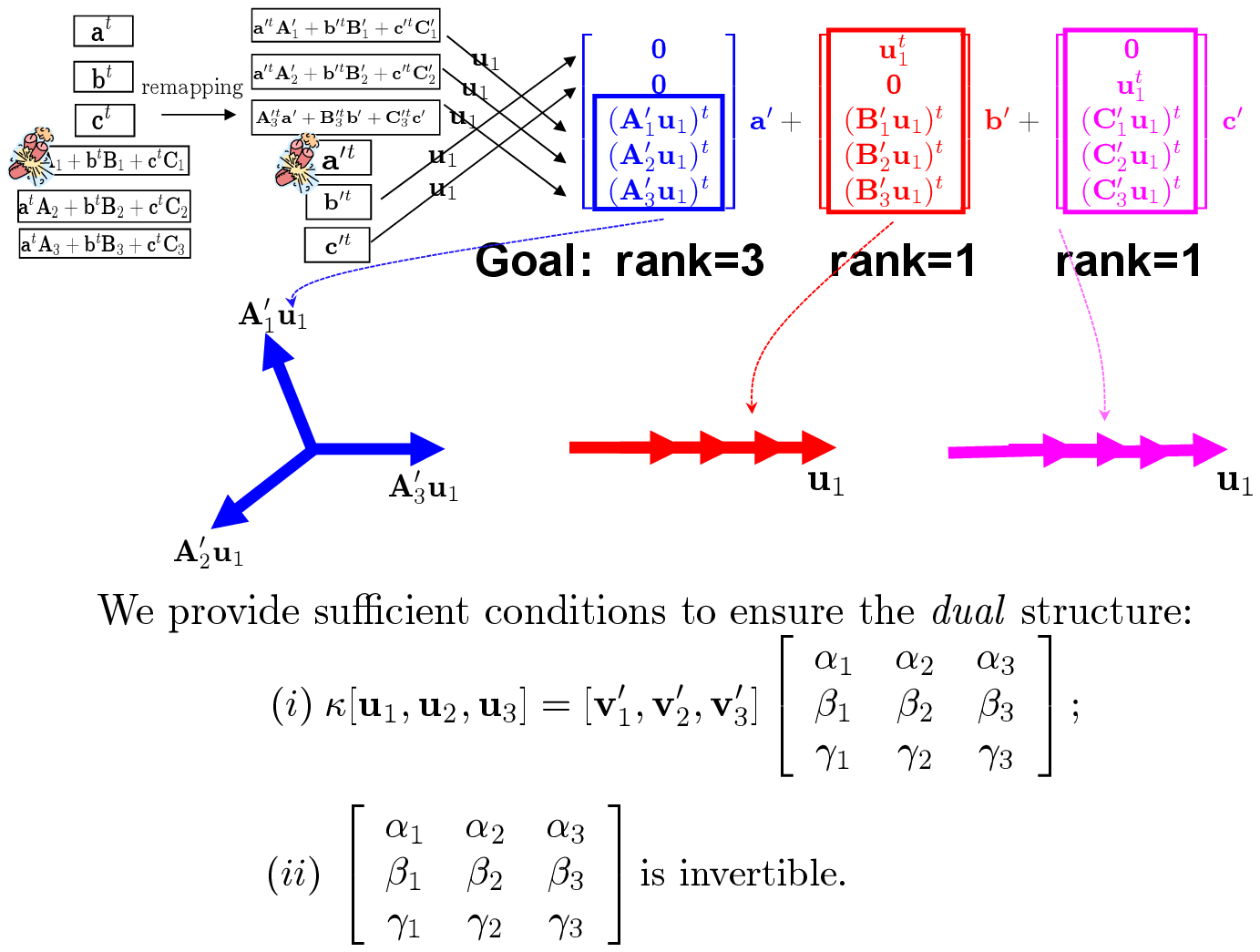, angle=0, width=0.85\textwidth}}
\end{center}
\caption{Exact repair of a parity node for $(6,3)$ E-MSR code. The idea is to construct the \emph{dual} code-structure of (\ref{eq-63_PrimeEM}) by remapping parity nodes and then adding sufficient conditions of (\ref{eq-63-MinvM}) and (\ref{eq:uiconditions}).}\label{fig:63_EMSR_parity}
\end{figure}

\begin{lemma}
\label{lemma:TensorElementaryMatrix}
Suppose
\begin{align}
\label{eq-63-MinvM}
\mathbf{M}: = \left[ \begin{array}{c|c|c}
    \alpha_{1} & \alpha_{2} & \alpha_{3}  \\
     \beta_{1} & \beta_{2} &  \beta_{3} \\
    \gamma_{1} & \gamma_{2} & \gamma_{3} \\
  \end{array} \right] \textrm{ is invertible}.
\end{align}
Also assume
\begin{align}
\label{eq:uiconditions}
\kappa \mathbf{U} = \mathbf{V}' \mathbf{M}.
\end{align}
where $\mathbf{U} = [\mathbf{u}_1, \mathbf{u}_2, \mathbf{u}_3]$,
$\mathbf{V}' = [\mathbf{v}_1', \mathbf{v}_2', \mathbf{v}_3']$, $\{ \mathbf{v}' \} := \{ \mathbf{v}_1', \mathbf{v}_2', \mathbf{v}_3' \}$ is the dual basis of $\{ \mathbf{v} \} $, i.e., $\mathbf{v}_i'^t \mathbf{v}_j = \delta(i-j)$ and $\kappa$ is an  arbitrary non-zero value s.t. $1-\kappa^2 \neq 0$.
Then, we can obtain the dual structure of (\ref{eq:63_EncodingMatrices}) as follows:
\begin{align}
\begin{split}
\label{eq-63_PrimeEM}
\mathbf{A}_1' &= \frac{1}{1-\kappa^2} \left(  \mathbf{v}_1' \mathbf{u}_1'^t  - \kappa^2 \alpha_1' \mathbf{I} \right);
\mathbf{B}_1' = \frac{1}{1-\kappa^2} \left( \mathbf{v}_1' \mathbf{u}_2'^t - \kappa^2 \alpha_2' \mathbf{I}\right) \; ;
\mathbf{C}_1' = \frac{1}{1-\kappa^2}\left(  \mathbf{v}_1' \mathbf{u}_3'^t - \kappa^2 \alpha_3' \mathbf{I} \right)\\
\mathbf{A}_2' &= \frac{1}{1-\kappa^2}\left( \mathbf{v}_2' \mathbf{u}_1'^t  - \kappa^2 \beta_1' \mathbf{I}\right); \;
\mathbf{B}_2' =  \frac{1}{1-\kappa^2}\left( \mathbf{v}_2' \mathbf{u}_2'^t - \kappa^2 \beta_2' \mathbf{I}\right); \;
\mathbf{C}_2' = \frac{1}{1-\kappa^2}\left( \mathbf{v}_2' \mathbf{u}_3'^t - \kappa^2 \beta_3' \mathbf{I} \right)\\
\mathbf{A}_3' &= \frac{1}{1-\kappa^2}\left( \mathbf{v}_3' \mathbf{u}_1'^t - \kappa^2 \gamma_1' \mathbf{I}\right); \;
\mathbf{B}_3' =  \frac{1}{1-\kappa^2}\left( \mathbf{v}_3' \mathbf{u}_2'^t - \kappa^2 \gamma_2' \mathbf{I}\right);\;
\mathbf{C}_3' =\frac{1}{1-\kappa^2} \left( \mathbf{v}_3' \mathbf{u}_3'^t - \kappa^2 \gamma_3' \mathbf{I}\right),
\end{split}
\end{align}
where  $\{ \mathbf{u}' \}$ is the dual basis of $\{ \mathbf{u} \}$, i.e., $\mathbf{u}_i'^t \mathbf{u}_j = \delta(i-j)$ and $(\alpha_i', \beta_i', \gamma_i')$'s are the dual basis vectors, i.e.,
$<(\alpha_i', \beta_i', \gamma_i'),(\alpha_j, \beta_j, \gamma_j) > = \delta(i-j)$:
\begin{align}
\left[
  \begin{array}{ccc}
    \alpha_{1}' & \beta_{1}' & \gamma_{1}'  \\
\hline
     \alpha_{2}' &  \beta_{2}' &  \gamma_{2}' \\
\hline
    \alpha_{3}' & \beta_{3}'   & \gamma_{3}' \\
  \end{array}
\right]: = \left[ \begin{array}{c|c|c}
    \alpha_{1} & \alpha_{2} & \alpha_{3}  \\
     \beta_{1} & \beta_{2} &  \beta_{3} \\
    \gamma_{1} & \gamma_{2} & \gamma_{3} \\
  \end{array} \right]^{-1}.
\end{align}

\end{lemma}
\begin{proof}
See Appendix~\ref{appendix:lemma_TensorEM}.
\end{proof}

\begin{remarks}
The dual structure of (\ref{eq-63_PrimeEM}) now gives exact-repair solutions for parity nodes. For exact repair of parity node 1, we can use vector $\mathbf{u}_1$ (a common eigenvector of the $\mathbf{B}_i'$'s and $\mathbf{C}_i'$'s), since it enables simultaneous interference alignment for $\mathbf{b}'$ and $\mathbf{c}'$, while ensuring the decodability of $\mathbf{a}'$. See Fig.~\ref{fig:63_EMSR_parity}. Notice that more conditions of (\ref{eq-63-MinvM}) and (\ref{eq:uiconditions}) are added to ensure exact repair of all nodes, while these conditions were unnecessary for exact repair of systematic nodes only. Also note these are only sufficient conditions.
\end{remarks}

\begin{remarks}
Note that the dual structure of (\ref{eq-63_PrimeEM}) is quite similar to the primary structure of (\ref{eq:63_EncodingMatrices}). The only difference is that in the dual structure, $\{ \mathbf{u} \} $ and $\{ \mathbf{v} \}$ are interchanged to form a \emph{transpose-like} structure. This reveals insights into how to guarantee exact repair of parity nodes in a transparent manner.
\end{remarks}


\subsection{The MDS-Code Property}
\label{sec-MDScodeProperty}

The third part of the framework is to guarantee the MDS-code property, which allows us to identify specific constraints on the $(\alpha_i, \beta_i, \gamma_i)$'s and/or ($\{ \mathbf{v} \}$, $\{ \mathbf{u} \}$). Consider four cases, associated in the Data Collector (DC) who is intended in the source file data: (1) 3 systematic nodes; (2) 3 parity nodes; (3) 1 systematic and 2 parity nodes; (4) 1 systematic and 2 parity nodes.

The first is a trivial case. The second case has been already verified in the process of forming the dual code-structure of (\ref{eq-63_PrimeEM}). The invertibility condition of (\ref{eq-63-MinvM}) together with (\ref{eq:uiconditions}) suffices to ensure the invertibility of the composite matrix. The third case requires the invertibility of all of each encoding matrix. In this case, it is necessary that the $\alpha_i$'s, $\beta_i$'s and $\gamma_i$'s are non-zero values; otherwise, each encoding matrix has rank 1. Also the non-zero values together with (\ref{eq:uiconditions}) guarantee the invertibility of each encoding matrix. Under these conditions, for example, the inverse of $\mathbf{A}_1$ is well defined as:
\begin{align*}
\mathbf{A}_1^{-1} &= \frac{1}{\alpha_1} \left( - \frac{1}{ \alpha_1 + \mathbf{v}_1^t \mathbf{u}_1} \mathbf{u}_1 \mathbf{v}_1^t + \mathbf{I} \right) \\
& =\frac{1}{\alpha_1} \left( - \frac{\kappa}{ \alpha_1 (\kappa  + 1)} \mathbf{u}_1 \mathbf{v}_1^t + \mathbf{I} \right).
\end{align*}
where the second equality follows from $\mathbf{v}_1^t \mathbf{u}_1 = \frac{\alpha_1}{\kappa}$ due to (\ref{eq:uiconditions}).

The last case requires some non-trivial work. Consider a specific example where the DC connects to nodes (3,4,5). In this case, we first recover $\mathbf{c}$ from node 3 and  subtract the terms associated with $\mathbf{c}$ from nodes 4 and 5. We then get:
\begin{align}
\begin{split}
\left[
  \begin{array}{cc}
    \mathbf{a}^t &  \mathbf{b}^t\\
  \end{array}
\right]
\left[
  \begin{array}{cc}
    \mathbf{A}_1 & \mathbf{A}_2 \\
    \mathbf{B}_1 & \mathbf{B}_2 \\
  \end{array}
\right] = \left[
  \begin{array}{cc}
    \mathbf{a}^t &  \mathbf{b}^t\\
  \end{array}
\right]
\left[
  \begin{array}{cc}
    \mathbf{u}_1 \mathbf{v}_1^t  + \alpha_1 \mathbf{I} & \mathbf{u}_2 \mathbf{v}_1^t  + \alpha_2 \mathbf{I} \\
    \mathbf{u}_1 \mathbf{v}_2^t  + \beta_1 \mathbf{I} & \mathbf{u}_2 \mathbf{v}_2^t + \beta_2 \mathbf{I} \\
  \end{array}
\right].
\end{split}
\end{align}
Using a Gaussian elimination method, we show that the sub-composite matrix is invertible if
\begin{align}
\mathbf{M}_2:= \left[ \begin{array}{c|c}
    \alpha_{1} & \alpha_{2} \\
     \beta_{1} & \beta_{2} \\
  \end{array} \right] \textrm{ is invertible}.
\end{align}
Here is the Gaussian elimination method:
\begin{align*}
\left[
  \begin{array}{cc}
    \mathbf{u}_1 \mathbf{v}_1^t  + \alpha_1 \mathbf{I} & \mathbf{u}_2 \mathbf{v}_1^t  + \alpha_2 \mathbf{I} \\
    \mathbf{u}_1 \mathbf{v}_2^t  + \beta_1 \mathbf{I} & \mathbf{u}_2 \mathbf{v}_2^t + \beta_2 \mathbf{I} \\
  \end{array}
\right] \overset{(a)} \sim \left[
                             \begin{array}{cc}
                               \mathbf{v}_1^t + \alpha_1 \mathbf{u}_1'^t & \alpha_2  \mathbf{u}_1'^t \\
                               \mathbf{v}_2^t + \beta_1 \mathbf{u}_1'^t & \beta_2  \mathbf{u}_1'^t \\
                               \alpha_1 \mathbf{u}_2'^t & \mathbf{v}_1^t + \alpha_2 \mathbf{u}_2'^t \\
                               \beta_1 \mathbf{u}_2'^t & \mathbf{v}_2^t + \beta_2 \mathbf{u}_2'^t \\
                               \alpha_1 \mathbf{u}_3'^t & \alpha_2 \mathbf{u}_3'^t \\
                               \beta_1 \mathbf{u}_3'^t & \beta_2 \mathbf{u}_3'^t \\
                             \end{array}
                           \right] \overset{(b)} \sim \left[
                             \begin{array}{cc}
                               \alpha_1' \mathbf{v}_1^t + \beta_1' \mathbf{v}_2^t + \mathbf{u}_1'^t & \mathbf{0}^t \\
                               \alpha_2' \mathbf{v}_1^t + \beta_2' \mathbf{v}_2^t &  \mathbf{u}_1'^t \\
                                \mathbf{u}_2'^t & \alpha_1' \mathbf{v}_1^t + \beta_1' \mathbf{v}_2^t \\
                               \mathbf{0}^t & \alpha_2' \mathbf{v}_1^t + \beta_2' \mathbf{v}_2^t + \mathbf{u}_2'^t \\
                                \mathbf{u}_3'^t &  \mathbf{0}^t \\
                                \mathbf{0}^t & \mathbf{u}_3'^t \\
                             \end{array}
                           \right].
                           \end{align*}
where $(a)$ following from multiplying $[\mathbf{u}_1'^t, \mathbf{0}^t;\mathbf{0}^t, \mathbf{u}_1'^t;\mathbf{u}_2'^t, \mathbf{0}^t; \mathbf{0}^t,\mathbf{u}_2'^t;
\mathbf{u}_3'^t, \mathbf{0}^t; \mathbf{0}^t,\mathbf{u}_3'^t  ]$ to the left; $(b)$ follows from multiplying $[\mathbf{M}_2^{-1}, \mathbf{0}, \mathbf{0}; \mathbf{0}, \mathbf{M}_2^{-1}, \mathbf{0};\mathbf{0}, \mathbf{0},\mathbf{M}_2^{-1} ]$ to the left. Here $(\alpha_i', \beta_i')$'s are the dual basis vectors of $(\alpha_i, \beta_i)$'s. Note that the resulting matrix is invertible, since $\{\mathbf{u}'\}$ is a dual basis.

Considering the above 4 cases, the following condition together with (\ref{eq-63-MinvM}) and (\ref{eq:uiconditions}) suffices for guaranteeing the MDS-code property:
\begin{align}
\label{eq:conditionforMDS}
\textit{Any submatrix of } \mathbf{M} \textit{ of (\ref{eq-63-MinvM}) is invertible.}
\end{align}

\subsection{Code Construction with Finite-Field Alphabets}
The last part is to design $\mathbf{M}$ of (\ref{eq-63-MinvM}) and $\{\mathbf{v} \}:=(\mathbf{v}_1, \mathbf{v}_2, \mathbf{v}_3)$ in (\ref{eq:63_EncodingMatrices}) such that $\{\mathbf{v} \}$ is linearly independent and the conditions of (\ref{eq:uiconditions}) and (\ref{eq:conditionforMDS}) are satisfied. First, in order to guarantee (\ref{eq:conditionforMDS}), we can use a Cauchy matrix, as it was used for the code introduced in \cite{KumarRamchandran_MSR}.

\begin{definition}[A Cauchy Matrix \cite{Bernstein:Cauchy}]
A Cauchy matrix $\mathbf{M}$ is an $m \times n$ matrix with entries $m_{ij}$ in the form:
\begin{align*}
m_{ij} = \frac{1}{x_i - y_j}, \forall i=1, \cdots, m, j=1, \cdots n, x_i \neq y_j,
\end{align*}
where $x_i$ and $y_j$ are elements of a field and $\{x_i \}$ and $\{ y_j \}$ are injective sequences, i.e., elements of the sequence are distinct.
\end{definition}

The injective property of $\{x_i \}$ and $\{ y_j \}$ requires a finite field size of $2s$ for an $s \times s$ Cauchy matrix. Therefore, in our $(6,3,5)$ code example, the finite field size of 6 suffices. The field size condition for guaranteeing linear independence of $\{ \mathbf{v} \}$ is more relaxed.

\subsection{Summary}
Using the code structure of (\ref{eq:63_EncodingMatrices}) and the conditions of (\ref{eq-63-MinvM}), (\ref{eq:uiconditions}) and (\ref{eq:conditionforMDS}), we can now state the following theorem.

\begin{theorem}[$(6,3,5)$ E-MSR Codes]
\label{theorem-63}
Suppose $\mathbf{M}$ of (\ref{eq-63-MinvM}) is a Cauchy matrix, i.e., every submatrix of is invertible.
Each element of $\mathbf{M}$ is in
${\sf GF} (q)$ and $q \geq 6$. Suppose encoding matrices form the code structure of (\ref{eq:63_EncodingMatrices}), $\{ \mathbf{v} \}:=(\mathbf{v}_1, \mathbf{v}_2, \mathbf{v}_3)$ is linearly independent, and $\{ \mathbf{u} \}$ satisfies the condition of (\ref{eq:uiconditions}).  Then, the code satisfies the MDS property and achieves the MSR point under exact repair constraints of \emph{all} nodes.
\end{theorem}

\begin{remarks}
Note that the code introduced in \cite{KumarRamchandran_MSR} is a special case of Theorem~\ref{theorem-63}, where $\mathbf{V} = [\mathbf{v}_1, \mathbf{v}_2, \mathbf{v}_3 ] = \mathbf{I}$. 
\end{remarks}

\section{Examples}
\label{section-(63)}

We provide two numerical examples: (1) an orthogonal code example where $\mathbf{V}=[\mathbf{v}_1, \mathbf{v}_2, \mathbf{v}_3]$ is orthogonal, e.g., $\mathbf{V}=\mathbf{I}$; (2) an bi-orthogonal code example where  $\mathbf{V}$ is not orthogonal but invertible. As mentioned earlier, the code in \cite{KumarRamchandran_MSR} belongs to the case of $\mathbf{V}=\mathbf{I}$.

We will also discuss the complexity of repair construction schemes for each of these examples. It turns out that the first code has significantly lower complexity for exact repair of systematic nodes, as compared to that of parity nodes. On the other hand, for the second bi-orthogonal codes, the specific choice of $\mathbf{V}=\kappa^{-1} \mathbf{M}^t$ gives $\mathbf{U}= \mathbf{I}$, thereby providing much simpler parity-node repair schemes instead. Depending on applications of interest, one can choose an appropriate code among our generalized family of codes.

\subsection{Orthogonal Case}

We present an example of $(6,3,5)$ E-MSR codes defined over ${\sf GF}(4)$ where $\mathbf{V}=\mathbf{I}$ and
\begin{align*}
\mathbf{M} = \left[
               \begin{array}{ccc}
                 1 & 1 & 1 \\
                 1 & 2 & 3 \\
                 1 & 3 & 2 \\
               \end{array}
             \right], \mathbf{U} = \kappa^{-1} \mathbf{V}' \mathbf{M} = 2 \left[
               \begin{array}{ccc}
                 1 & 1 & 1 \\
                 1 & 2 & 3 \\
                 1 & 3 & 2 \\
               \end{array}
             \right] = \left[
               \begin{array}{ccc}
                 2 & 2 & 2 \\
                 2 & 3 & 1 \\
                 2 & 1 & 3 \\
               \end{array}
             \right],
\end{align*}
where $\mathbf{U}$ is set based on (\ref{eq:uiconditions}) and $\kappa=2^{-1}$. We use a generator polynomial of $g(x)=x^2+x+1$. Notice that we employ a \emph{non-Cauchy}-type matrix to construct a field-size 4 code (smaller than 6 required when using a Cauchy matrix). Remember that a Cauchy matrix provides only a sufficient condition for ensuring the invertibility of any submatrices of $\mathbf{M}$. By (\ref{eq:63_EncodingMatrices}) and (\ref{eq-63_PrimeEM}), the primary and dual code structures are given by

\begin{align}
\label{eq:63GGp_Orthogonal}
 \mathbf{G}= \left[
  \begin{array}{ccc|ccc|ccc}
    3 & 0   &0  &3  &0  &0  &3  &0  &0\\
    2 & 1   &0  &3  &1  &0  &1  &1  &0\\
    2 & 0   &1  &1  &0  &1  &3  &0  &1\\
\hline
    1 & 2   &0  &2  &2  &0  &3  &2  &0\\
    0 & 3   &0  &0  &1  &0  &0  &2  &0\\
    0 & 2   &1  &0  &1  &2  &0  &3  &3\\
\hline
    1 & 0   &2  &3  &0  &2  &2  &0  &2\\
    0 & 1   &2  &0  &3  &3  &0  &2  &1\\
    0 & 0   &3  &0  &0  &2  &0  &0  &1\\
  \end{array}
\right]; \mathbf{G}^{-1}= \left[
  \begin{array}{ccc|ccc|ccc}
    2 & 1   &1  &3  &0  &0  &3  &0  &0\\
    0 & 3   &0  &1  &2  &1  &0  &3  &0\\
    0 & 0   &3  &0  &0  &3  &1  &1  &2\\
\hline
    2 & 3   &2  &2  &0  &0  &1  &0  &0\\
    0 & 3   &0  &1  &1  &2  &0  &1  &0\\
    0 & 0   &3  &0  &0  &2  &1  &3  &3\\
\hline
    2 & 2   &3  &1  &0  &0  &2  &0  &0\\
    0 & 3   &0  &1  &3  &3  &0  &2  &0\\
    0 & 0   &3  &0  &0  &1  &1  &2  &1\\
  \end{array}
\right].
\end{align}
where
\begin{align*}
\mathbf{G}:= \left[
   \begin{array}{ccc}
     \mathbf{A}_1 & \mathbf{A}_2 & \mathbf{A}_3 \\
      \mathbf{B}_1 & \mathbf{B}_2 & \mathbf{B}_3\\
      \mathbf{C}_1 & \mathbf{C}_2 & \mathbf{C}_3 \\
   \end{array}
 \right]; \mathbf{G}^{-1} =\left[
   \begin{array}{ccc}
     \mathbf{A}_1' & \mathbf{A}_2' & \mathbf{A}_3' \\
      \mathbf{B}_1' & \mathbf{B}_2' & \mathbf{B}_3'\\
      \mathbf{C}_1' & \mathbf{C}_2' & \mathbf{C}_3' \\
   \end{array}
 \right].
\end{align*}

\begin{figure}[t]
\begin{center}
{\epsfig{figure=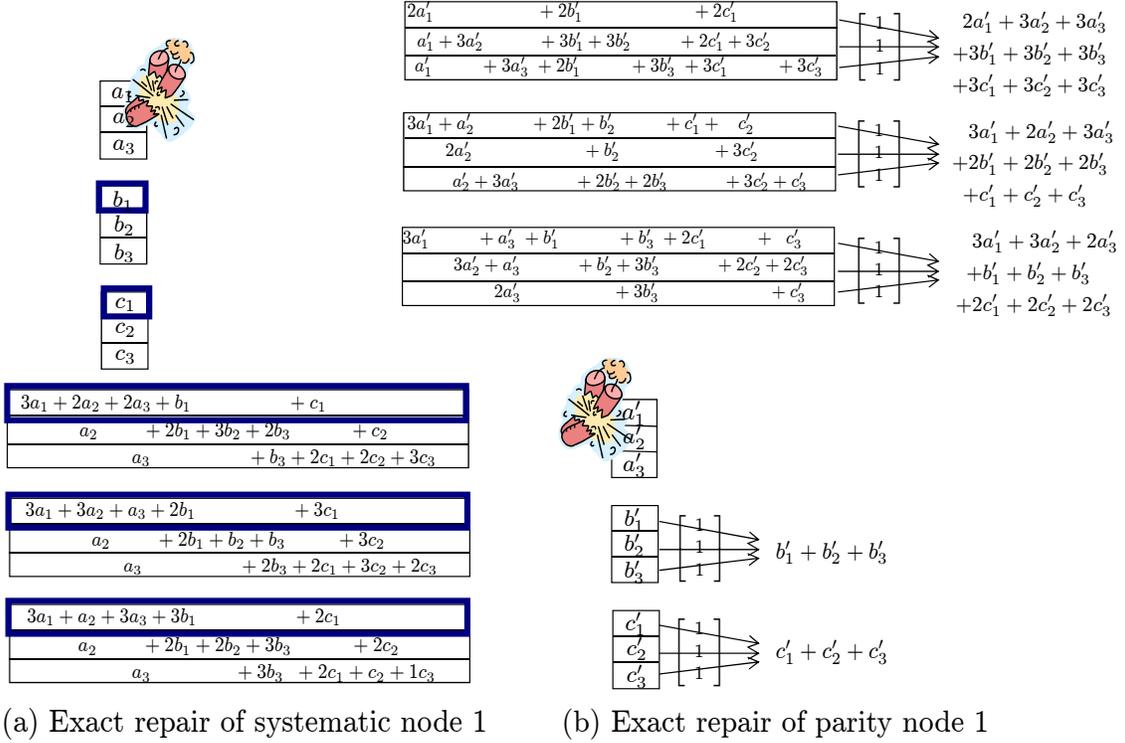, angle=0, width=0.9\textwidth}}
\end{center}
\caption{Orthogonal case: Illustration of exact pair for a $(6,3,5)$ E-MSR code defined over ${\sf GF}(4)$ where a generator polynomial $g(x)= x^2 + x +1$. The projection vector solution for systematic node repair is quite simple: $\mathbf{v}_{\alpha i}= \mathbf{v}_1=(1,0,0)^t, \forall i$. 
We download only the first equation from each survivor node; For parity node repair, our new framework provides a simple scheme: setting all of the projection vectors as $ 2^{-1} \mathbf{u}_1 = (1,1,1)^t$. This enables simultaneous interference alignment, while guaranteeing the decodability of $\mathbf{a}$.} \label{fig:63_EMSR_Orthogonal}
\end{figure}

Fig.~\ref{fig:63_EMSR_Orthogonal} shows an example for exact repair of $(a)$ systematic node 1 and $(b)$ parity node 1. Note that the projection vector solution for systematic node repair is quite simple: $\mathbf{v}_{\alpha i}= \mathbf{v}_1=(1,0,0)^t, \forall i$. We download only the first equation from each survivor node. Notice that the downloaded five equations contain only five unknown variables of $(a_1,a_2,a_3, b_1,c_1)$ and three equations associated with $\mathbf{a}$ are linearly independent. Hence, we can successfully recover $\mathbf{a}$.

On the other hand, exact repair of parity nodes seems non-straightforward. However, our framework provides quite a simple repair scheme: setting all of the projection vectors as $2^{-1} \mathbf{u}_1 = (1,1,1)^t$. This enables simultaneous interference alignment, while guaranteeing the decodability of $\mathbf{a}$. Notice that $(b_1',b_2',b_3')$ and  $(c_1',c_2',c_3')$ are aligned into $b_1'+b_2'+b_3'$ and $c_1'+c_2'+c_3'$, respectively, while three equations associated with $\mathbf{a}'$ are linearly independent.

As one can see, the complexity of systematic node repair is a little bit lower than that of parity node repair, although both repair schemes are simple. Hence, one can expect that this orthogonal code is useful for the applications where the complexity of systematic node repair needs to be significantly low.

\subsection{Bi-Orthogonal Case}

We provide another example of $(6,3,5)$ E-MSR codes where $\mathbf{V}$ is not orthogonal but invertible. We use the same field size of 4, the same generator polynomial and the same $\mathbf{M}$. Instead we choose non-orthogonal $\mathbf{V}$ so that the complexity of parity node repair can be significantly low. Our framework provides a concrete guideline for designing this type of code. Remember that the projection vector solutions are $\mathbf{u}_1$, $\mathbf{u}_2$ and $\mathbf{u}_3$ for exact repair of each parity node, respectively. For low complexity, we can first set $\mathbf{U} = \mathbf{I}$. The condition (\ref{eq:uiconditions}) then gives the following choice:
\begin{align*}
\mathbf{V} = \mathbf{M}^t \kappa^{-1} = \left[
               \begin{array}{ccc}
                 2 & 2 & 2 \\
                 2 & 3 & 1 \\
                 2 & 1 & 3 \\
               \end{array}
             \right],
\end{align*}
where we use $\kappa=2^{-1}$. By (\ref{eq:63_EncodingMatrices}) and (\ref{eq-63_PrimeEM}), the primary and dual code structures are given by

\begin{align}
\label{eq:63GGp_BiOrthogonal}
 \mathbf{G}= \left[
  \begin{array}{ccc|ccc|ccc}
    3 & 2   &2  &1  &0  &0  &1  &0  &0\\
    0 & 1   &0  &2  &3  &2  &0  &1  &0\\
    0 & 0   &1  &0  &0  &1  &2  &2  &3\\
\hline
    3 & 3   &1  &2  &0  &0  &3  &0  &0\\
    0 & 1   &0  &2  &1  &1  &0  &3  &0\\
    0 & 0   &1  &0  &0  &2  &2  &3  &2\\
\hline
    3 & 1   &3  &3  &0  &0  &2  &0  &0\\
    0 & 1   &0  &2  &2  &3  &0  &2  &0\\
    0 & 0   &1  &0  &0  &3  &2  &1  &1\\
  \end{array}
\right]; \mathbf{G}^{-1}= \left[
  \begin{array}{ccc|ccc|ccc}
    2 & 0   &0  &2  &0  &0  &2  &0  &0\\
    1 & 3   &0  &3  &3  &0  &2  &3  &0\\
    1 & 0   &3  &2  &0  &3  &3  &0  &3\\
\hline
    3 & 1   &0  &2  &1  &0  &1  &1  &0\\
    0 & 2   &0  &0  &1  &0  &0  &3  &0\\
    0 & 1   &3  &0  &2  &2  &0  &3  &1\\
\hline
    3 & 0   &1  &1  &0  &1  &2  &0  &1\\
    0 & 3   &1  &0  &1  &3  &0  &2  &2\\
    0 & 0   &2  &0  &0  &2  &0  &0  &1\\
  \end{array}
\right].
\end{align}
Notice that the matrices of (\ref{eq:63GGp_BiOrthogonal}) have exactly the transpose structure of the matrices of (\ref{eq:63GGp_Orthogonal}). Hence, this code of (\ref{eq:63GGp_BiOrthogonal}) is a dual solution of (\ref{eq:63GGp_Orthogonal}), thereby providing switched projection vector solutions and lowering the complexity for parity node repair.

\begin{figure}[t]
\begin{center}
{\epsfig{figure=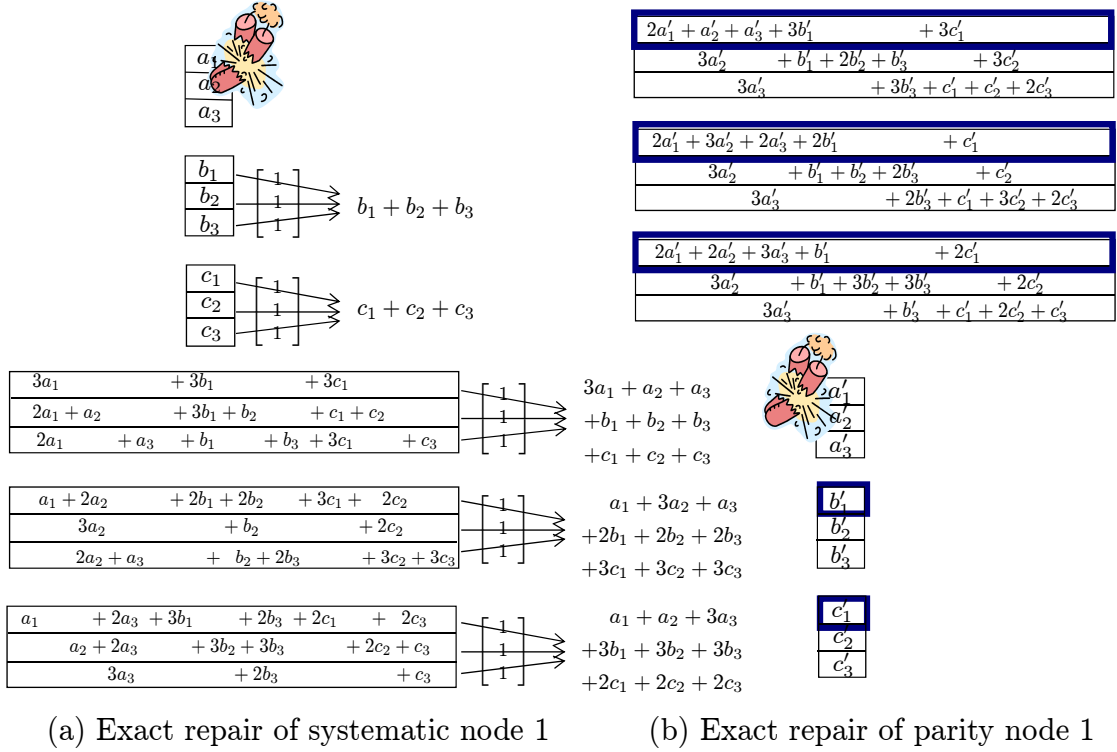, angle=0, width=0.9\textwidth}}
\end{center}
\caption{Bi-Orthogonal case: Illustration of exact repair for a $(6,3,5)$ E-MSR code defined over ${\sf GF}(4)$ where a generator polynomial $g(x)= x^2 + x +1$. We use $\mathbf{U}=\mathbf{I}$. For parity node repair, the solution for projection vectors is much simpler. We download only the first equation from each survivor node; Systematic node repair is a bit involved: setting all of the projection vectors as $2^{-1} \mathbf{v}_1 = (1,1,1)^t$.} \label{fig:63_EMSR_BiOrthogonal}
\end{figure}

Fig.~\ref{fig:63_EMSR_BiOrthogonal} shows an example for exact repair of $(a)$ systematic node 1 and $(b)$ parity node 1. Reverse to the previous case, exact repair of parity nodes is now much simpler. In this example, by downloading only the first equation from each survivor node, we can successfully recover $\mathbf{a}'$. On the contrary, systematic node repair is a bit involved: a projection vector solution is $2^{-1} \mathbf{v}_1=(1,1,1)^t$. Using this vector, we can achieve simultaneous interference alignment, thereby decoding the desired components of $\mathbf{a}'$.

\section{Generalization: $n \geq 2k; d \geq 2k-1 $}
\label{sec-generalization}

Theorem~\ref{theorem-63} gives insights into generalization to $(2k,k,k-1)$ E-MSR codes. The key observation is that assuming $\mathcal{M}=k(d-k+1)$, storage cost is $\alpha = \mathcal{M}/k=d-k+1=k$ and this number is equal to the number of systematic nodes and furthermore matches the number of parity nodes. Notice that the storage size matches the size of encoding matrices, which determines the number of linearly independent vectors of $\{\mathbf{v} \}:= \{ \mathbf{v}_1, \cdots \}$.
In this case, therefore, we can generate $k$ linearly independent vectors $\{\mathbf{v} \}:= \{ \mathbf{v}_1, \cdots, \mathbf{v}_k \}$ and corresponding $\{\mathbf{u} \}:= \{ \mathbf{u}_1, \cdots, \mathbf{u}_k \}$ through the appropriate choice of $\mathbf{M}$ to design $(2k,k,k-1)$ E-MSR codes.

\subsection{Case: $n = 2k$}

\begin{theorem}[$(2k,k,k-1)$ E-MSR Codes]
\label{theorem-2kk}
Let $\mathbf{M}$ be a Cauchy matrix:
\begin{align*}
\mathbf{M} = \left[
                \begin{array}{cccc}
                  m_1^{(1)} & m_1^{(2)} & \cdots    & m_{1}^{(k)} \\
                  m_2^{(1)} & m_2^{(2)} & \cdots    & m_{2}^{(k)} \\
                  \vdots    & \vdots    & \ddots    & \vdots \\
                  m_k^{(1)} & m_k^{(2)} & \cdots    & m_{k}^{(k)} \\
                \end{array}
              \right],
\end{align*}
where each element $m_{j}^{(i)}\in {\sf GF} (q)$, where $q \geq 2k$. Suppose
\begin{align}
\begin{split}
&\mathbf{V} = [\mathbf{v}_1, \cdots, \mathbf{v}_k] \textrm{ is invertible and} \\
 &\mathbf{U}= \kappa^{-1} \mathbf{V}' \mathbf{M},
 \end{split}
\end{align}
  where $\mathbf{V}'= (\mathbf{V}^t)^{-1}$ and $\kappa$ is an arbitrary non-zero value $\in \mathbb{F}_q$ such that $1-\kappa^2 \neq 0$. Also assume that encoding matrices are given by
\begin{align}
\begin{split}
\mathbf{G}_{1}^{(1)} &= \mathbf{u}_{1} \mathbf{v}_1^t + m_1^{(1)} \mathbf{I},
\cdots,
\mathbf{G}_{k}^{(1)} = \mathbf{u}_{1} \mathbf{v}_k^t + m_k^{(1)} \mathbf{I}, \\
&\vdots \qquad \qquad \qquad  \ddots \qquad \qquad \qquad  \vdots \\
\mathbf{G}_{1}^{(k)} &= \mathbf{u}_{k} \mathbf{v}_1^t + m_1^{(k)} \mathbf{I},
\cdots,
\mathbf{G}_{k}^{(k)} = \mathbf{u}_{k} \mathbf{v}_k^t + m_k^{(k)} \mathbf{I}, \\
\end{split}
\end{align}
where $\mathbf{G}_l^{(i)}$ indicates an encoding matrix for parity node $i$, associated with information unit $l$.
Then, the code satisfies the MDS property and achieves the MSR point under exact repair constraints of all nodes.
\end{theorem}
\begin{proof}
See Appendix \ref{appen-Theorem2kk}.
\end{proof}

\begin{remarks}
Note that the minimum required alphabet size is $2k$. As mentioned earlier, this is because we employ a Cauchy matrix for ensuring the invertibility of any submatrices of $\mathbf{M}$. One may customize codes to find smaller alphabet-size codes.
\end{remarks}

\subsection{Case: $n \geq 2k; d \geq 2k-1$}

Now what if $k$ is less than the size $(=\alpha = d-k+1)$ of encoding matrices, i.e., $ d \geq 2k-1$? Note that this case automatically implies that $n \geq 2k$, since $n \geq d+1$. The key observation in this case is that the encoding matrix size is bigger than $k$, and therefore we have more degrees of freedom (a larger number of linearly independent vectors) than the number of constraints. Hence, exact repair of systematic nodes becomes transparent. This was observed as well in \cite{KumarRamchandran_MSR}, where it was shown that for this regime, exact repair of systematic nodes only can be guaranteed by judiciously manipulating $(2k,k,k-1)$ codes through a puncturing operation.

We show that the puncturing technique in \cite{KumarRamchandran_MSR} (meant for exact repair of systematic nodes and for a special case of our generalized codes) together with our repair construction schemes can also carry over to ensure exact repair of \emph{all} nodes even for the generalized family of codes.
The recipe for this has two parts:
\begin{enumerate}
\item[1.] Constructing a target code from a larger code through the puncturing technique.
\item[2.] Showing that the resulting target code indeed ensures exact repair of all nodes as well as the MDS-code property for our generalized family of codes.
\end{enumerate}

The first part contains the following detailed steps:
\begin{enumerate}
  \item[1($a$)] Using Theorem~\ref{theorem-2kk}, construct a larger $(2n-2k, n-k, 2n-2k-1)$ code with a finite field size of $q \geq 2n-2k$.
  \item[1($b$)] Remove all the elements associated with the $(n-2k)$ information units (e.g., from the $(k+1)$th to the $(n-k)$th information unit). The number of nodes is then reduced by $(n-2k)$ and so are the number of information units and the number of degrees. Hence, we obtain the $(n,k,n-1)$ code.
  \item[1($c$)] Prune the last $(n-1-d)$ equations in each storage node and also the last $(n-1-d)$ symbols of each information unit, while keeping the number of information units and storage nodes. We can then get the $(n,k,d)$ target code.
\end{enumerate}
Indeed, based on our framework in Section~\ref{sec-BasisFramework}, it can be shown that the resulting \emph{punctured} code described above guarantees exact repair of all nodes and MDS-code property for our generalized family of codes. Hence, we obtain the following theorem. The proof procedure is tedious and mimics that of Theorem~\ref{theorem-2kk}. Therefore, details are omitted.

\begin{theorem}[$\frac{k}{n} \leq \frac{1}{2}, d \geq 2k-1$]
\label{theorem-2kn}
Under exact repair constraints of all nodes, the optimal tradeoff of (\ref{eq-MSRpoint}) can be attained with a deterministic scheme requiring a field size of at most $2(n-k)$.
\end{theorem}

\begin{figure}[t]
\begin{center}
{\epsfig{figure=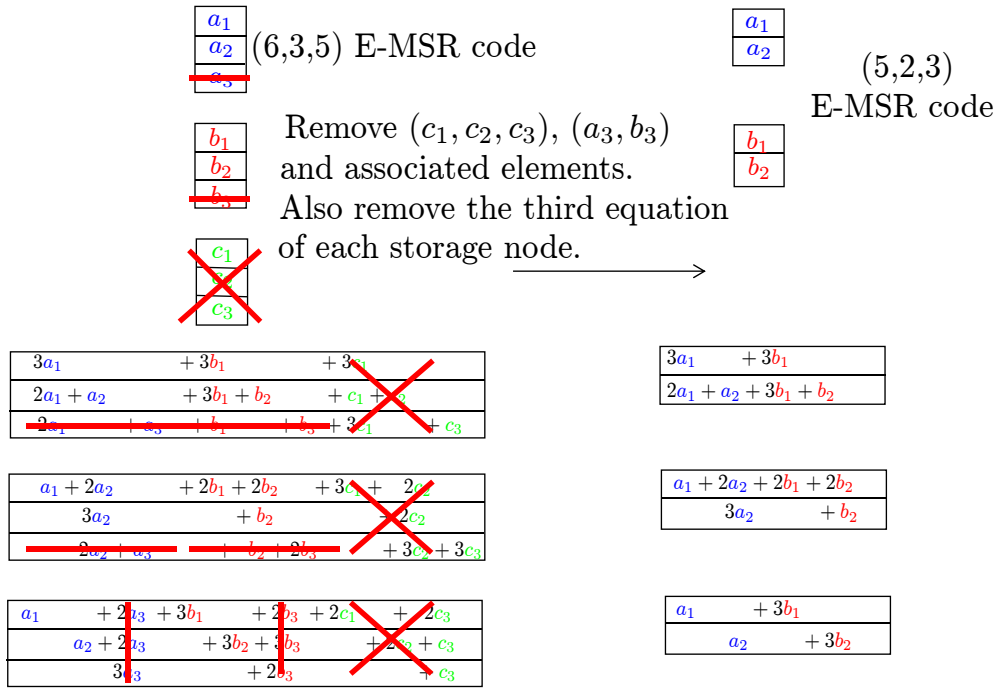, angle=0, width=0.8\textwidth}}
\end{center}
\caption{Bi-Orthogonal case: Illustration of the construction of a $(5,2,3)$ E-MSR code from a $(6,3,5)$ code defined over ${\sf GF}(4)$. For a larger code, we adopt the $(6,3,5)$ code in Fig.~\ref{fig:63_EMSR_BiOrthogonal}. First, we remove all the elements associated with the last $(n-2k)=1$ information unit (``$\mathbf{c}$''). Next, we prune symbols $(a_3,b_3)$ and associated elements. Also we remove the last equation of each storage node. Finally we obtain the $(n,k,d)=(5,2,3)$ target code.} \label{fig:63_52_convert_b}
\end{figure}

\begin{example}
Fig.~\ref{fig:63_52_convert_b} illustrates how to construct an $(n,k,d)=(5,2,3)$ target code based on the above recipe. First construct the $(2n-2k,n-k,2n-2k-1)=(6,3,5)$ code, which is larger than the $(5,2,3)$ target code, but which belongs to the category of $n=2k$. For this code, we adopt the bi-orthogonal case example in Fig.~\ref{fig:63_EMSR_BiOrthogonal}. For this code, we now remove all the elements associated with the last $(n-2k)=1$ information unit, which corresponds to $(c_1,c_2,c_3)$. Next, prune the last symbol $(a_3,b_3)$ of each information unit and associated elements to shrink the storage size into $2$. We can then obtain the $(5,2,3)$ target code. Exact repair and the MDS-code property of the resulting code can be verified based on the proposed framework in Section~\ref{sec-BasisFramework}.
\end{example}

\section{Generalization: $ k \leq 3$}

As a side generalization, we consider the case of $k \leq 3$. The interesting special case of the $(5,3)$ E-MSR code\footnote{Independently, the authors in \cite{Cullina_MSR} found $(5,3)$ codes defined over ${\sf GF}(3)$, based on a search algorithm.} will be focused on, since it is not covered by the above case of $\frac{k}{n} \leq \frac{1}{2}$. For this case, we propose another interference alignment technique building on an eigenvector concept.

\begin{theorem}[$k \leq 3$]
The MSR point can be attained with a deterministic scheme requiring a finite-field size of at most $2n-2k$.
\end{theorem}
\begin{proof}
The case of $k=1$ is trivial. By Theorems \ref{theorem-2kk} and \ref{theorem-2kn}, we prove the case of $k=2$. However, additional effort is needed to prove the case of $k=3$.
By Theorems \ref{theorem-2kk} and \ref{theorem-2kn}, $(n,3)$ for $n \geq 6$ can be proved. But  $(5,3)$ codes are not in the class. In Section~\ref{sec:53_E-MSR}, we will address this case to complete the proof.
\end{proof}

\begin{remarks} In order to cover general $n$, we provide a looser bound on the required finite-field size: $q \geq 2n-2k$. In fact, for the $(5,3)$ code (that will be shown in Lemma \ref{lemma-(53)}), a smaller finite-field size of $q=3$ ($<4=2n-2k$) is enough for construction. We have taken the maximum of the required field sizes of all the cases.
\end{remarks}

\subsection{$(5,3)$ E-MSR Codes}
\label{sec:53_E-MSR}

We consider $d=4$ and $\mathcal{M}=6$. The cutset bound (\ref{eq-MSRpoint}) then gives the fundamental limits of: storage cost $\alpha=2$ and repair-bandwidth-per-link=1; hence, the dimension of encoding matrices is 2-by-2. Note that the size is \emph{less} than the number of systematic nodes. Therefore, our earlier framework does not cover this category. In fact, the $(5,3)$ code is in the case of $n+1 = 2k$, where it was shown in \cite{KumarRamchandran_MSR} that there exist codes that achieve the cutset bound under exact repair of systematic nodes only (not including parity nodes).

We propose an eigenvector-based interference alignment technique to prove the code existence under exact repair of \emph{all} nodes. Let $\mathbf{a}=(a_1,a_2)^t$, $\mathbf{b} = (b_1,b_2)^t$ and $\mathbf{c}=(c_1,c_2)^t$. For exact repair, we connect to $4 (=d)$ nodes to download a one-dimensional scalar value from each node.
Fig.~\ref{fig:53example} illustrates exact repair of node 1. We download four equations from survivor nodes: $\mathbf{b}^{t} \mathbf{v}_{\alpha 1}$;
$\mathbf{c}^{t} \mathbf{v}_{\alpha 2}$; $\mathbf{a}^t (\mathbf{A}_1 \mathbf{v}_{\alpha 3}) +\mathbf{b}^t (\mathbf{B}_1  \mathbf{v}_{\alpha 3}) +\mathbf{c}^t (\mathbf{C}_1  \mathbf{v}_{\alpha 3})$; $\mathbf{a}^t ( \mathbf{A}_2 \mathbf{v}_{\alpha 4} ) + \mathbf{b}^t ( \mathbf{B}_2  \mathbf{v}_{\alpha 4} ) +\mathbf{c}^t (\mathbf{C}_2  \mathbf{v}_{\alpha 4})$.
The approach is different from that of our earlier proposed framework. Instead an idea here consists of three steps: (1) choosing projection vectors for achieving interference alignment; (2) gathering all the alignment constraints and the MDS-code constraint; (3) designing the encoding matrices that satisfy all the constraints. Notice the design of encoding matrices is the last part.

Here are details. Note that there are 6 unknown variables: 2 desired unknowns $(a_1,a_2)$ and 4 undesired unknowns $(b_1,b_2,c_1,c_2)$. Therefore, it is required to align $(b_1,b_2, c_1,c_2)$ onto at least 2-dimensional linear space. We face the challenge that appeared in the $(6,3,5)$ code example in Fig.~\ref{fig:63_EMSR_Challenge}. Projection vectors $\mathbf{v}_{\alpha3}$ and $\mathbf{v}_{\alpha 4}$ affect interference alignment $\mathbf{b}$ and $\mathbf{c}$ simultaneously.
Therefore, we need simultaneous interference alignment.
\begin{figure}[t]
\begin{center}
{\epsfig{figure=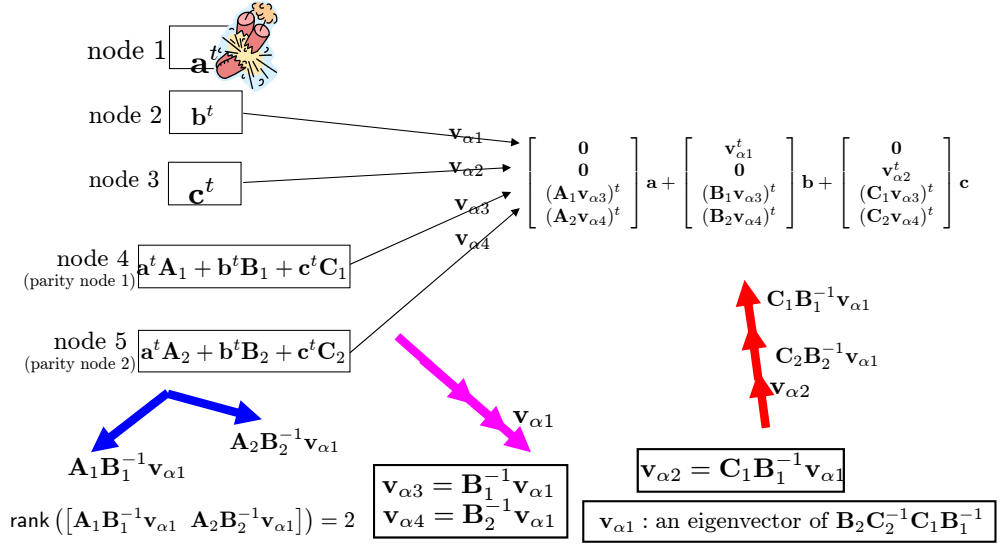, angle=0, width=0.8\textwidth}}
\end{center}
\caption{Eigenvector-based interference alignment for $(5,3)$ E-MSR codes. First we align interference ``$\mathbf{b}$'' by setting $\mathbf{v}_{\alpha 3} = \mathbf{B}_1^{-1} \mathbf{v}_{\alpha 1}$ and $\mathbf{v}_{\alpha 4} = \mathbf{B}_2^{-1} \mathbf{v}_{\alpha 1}$. Next, partially align interference of ``$\mathbf{c}$'' by setting $\mathbf{v}_{\alpha 2} = \mathbf{C}_1 \mathbf{B}_1^{-1} \mathbf{v}_{\alpha 1}$. Finally, choosing $\mathbf{v}_{\alpha 1}$ as an eigenvector of $\mathbf{B}_{2} \mathbf{C}_2^{-1} \mathbf{C}_1 \mathbf{B}_1^{-1}$, we can achieve interference alignment for $\mathbf{c}$.} \label{fig:53example}
\end{figure}
To solve this problem, we introduce an eigenvector-based interference alignment scheme.

First choose $\mathbf{v}_{\alpha 3}$ and $\mathbf{v}_{\alpha 4}$ such that $\mathbf{v}_{\alpha 3} = \mathbf{B}_1^{-1} \mathbf{v}_{\alpha 1}$ and $\mathbf{v}_{\alpha 4} = \mathbf{B}_2^{-1} \mathbf{v}_{\alpha 1}$, thereby achieving interference alignment for ``$\mathbf{b}$''. Observe the interfering vectors associated with ``$\mathbf{c}$'':
\begin{align*}
\mathbf{v}_{\alpha 2}; \;\mathbf{C}_1 \mathbf{B}_1^{-1} \mathbf{v}_{\alpha 1}; \; \mathbf{C}_2 \mathbf{B}_2^{-1} \mathbf{v}_{\alpha 1}.
\end{align*}
The first and second vectors can be aligned by setting $ \mathbf{v}_{\alpha 2} = \mathbf{C}_1 \mathbf{B}_1^{-1} \mathbf{v}_{\alpha 1}$. Now what about for the following two vectors: $\mathbf{C}_1 \mathbf{B}_1^{-1} \mathbf{v}_{\alpha 1}$ and $\mathbf{C}_2 \mathbf{B}_2^{-1}  \mathbf{v}_{\alpha 1}$?
Suppose that the associated matrices ($\mathbf{C}_1 \mathbf{B}_1^{-1}$ and $\mathbf{C}_2 \mathbf{B}_2^{-1}$) and the projection vector $\mathbf{v}_{\alpha 1}$ are randomly chosen. Then, the these two vectors are not guaranteed to be aligned. However, a judicious choice of $\mathbf{v}_{\alpha 1}$ makes it possible to align them. The idea is to choose $\mathbf{v}_{\alpha 1}$ as an eigenvector of
$\mathbf{B}_2 \mathbf{C}_2^{-1} \mathbf{C}_1 \mathbf{B}_1^{-1}$. Since $\mathbf{v}_{\alpha 1}$ can be chosen arbitrarily, this can be easily done.
Lastly consider the condition for ensuring the decodability of desired signals: $\mathsf{rank} \left( \left[ \mathbf{A}_1 \mathbf{B}_1^{-1} \mathbf{v}_{\alpha 1} \;\; \mathbf{A}_2  \mathbf{B}_2^{-1} \mathbf{v}_{\alpha 1} \right] \right) = 2$.

We repeat the procedure for exact repair of ``$\mathbf{b}$'' and ``$\mathbf{c}$''. For parity nodes, we employ the remapping technique described earlier:
\begin{align}
\left[
  \begin{array}{ccc}
   \mathbf{a}' \\
    \mathbf{b}' \\
    \mathbf{c}' \\
  \end{array}
\right]:=\left[
   \begin{array}{ccc}
     \mathbf{A}_1^t  & \mathbf{B}_1^t & \mathbf{C}_1^t \\
     \mathbf{A}_2^t  & \mathbf{B}_2^t & \mathbf{C}_2^t \\
\mathbf{0} & \mathbf{0} & \mathbf{I} \\
   \end{array}
 \right] \left[
  \begin{array}{ccc}
   \mathbf{a} \\
    \mathbf{b} \\
    \mathbf{c} \\
  \end{array}
\right], \left[
    \begin{array}{ccc}
      \mathbf{A}_1'  & \mathbf{A}_2' & \mathbf{0}\\
      \mathbf{B}_1'  & \mathbf{B}_2' & \mathbf{0}\\
       \mathbf{C}_1' & \mathbf{C}_2' &  \mathbf{I} \\
    \end{array}
  \right]:=\left[
    \begin{array}{ccc}
      \mathbf{A}_1  & \mathbf{A}_2 & \mathbf{0}\\
      \mathbf{B}_1  & \mathbf{B}_2 & \mathbf{0}\\
       \mathbf{C}_1 & \mathbf{C}_2 &  \mathbf{I} \\
    \end{array}
  \right]^{-1}.
\end{align}
We gather all the conditions that need to be guaranteed for exact repair of all nodes:
\begin{align}
\begin{split}
\label{eq-53conditions}
& \mathsf{rank} \left( \left[ \mathbf{A}_1 \mathbf{B}_1^{-1} \mathbf{v}_{\alpha 1} \;\;
\mathbf{A}_2  \mathbf{B}_2^{-1} \mathbf{v}_{\alpha 1} \right] \right) = 2, \\
& \mathsf{rank} \left( \left[ \mathbf{B}_1 \mathbf{C}_1^{-1} \mathbf{v}_{\beta 1} \;\;
\mathbf{B}_2  \mathbf{C}_2^{-1} \mathbf{v}_{\beta 1} \right] \right) = 2, \\
& \mathsf{rank} \left( \left[ \mathbf{C}_1 \mathbf{A}_1^{-1} \mathbf{v}_{\gamma 1} \;\;
\mathbf{C}_2  \mathbf{A}_2^{-1} \mathbf{v}_{\gamma 1} \right] \right) = 2, \\
& \mathsf{rank} \left( \left[ \mathbf{A}_1' \mathbf{B}_1'^{-1} \mathbf{v}_{\alpha' 1} \;\;
\mathbf{A}_2'  \mathbf{B}_2'^{-1} \mathbf{v}_{\alpha' 1} \right] \right) = 2, \\
& \mathsf{rank} \left( \left[ \mathbf{B}_1' \mathbf{C}_1'^{-1} \mathbf{v}_{\beta' 1} \;\;
\mathbf{B}_2'  \mathbf{C}_2'^{-1} \mathbf{v}_{\beta' 1} \right] \right) = 2, \\
\end{split}
\end{align}
where
\begin{align}
\begin{split}
\label{eq-53conditions3}
& \mathbf{v}_{\alpha 1}: \textrm{an eigenvector of }\mathbf{B}_2 \mathbf{C}_2^{-1} \mathbf{C}_1 \mathbf{B}_1^{-1}, \\
& \mathbf{v}_{\beta 1}: \textrm{an eigenvector of }\mathbf{C}_2 \mathbf{A}_2^{-1} \mathbf{A}_1 \mathbf{C}_1^{-1}, \\
& \mathbf{v}_{\gamma 1}: \textrm{an eigenvector of }\mathbf{A}_2 \mathbf{B}_2^{-1} \mathbf{B}_1 \mathbf{A}_1^{-1}, \\
& \mathbf{v}_{\alpha' 1}: \textrm{an eigenvector of }\mathbf{B}_2' \mathbf{C}_2'^{-1} \mathbf{C}_1' \mathbf{B}_1'^{-1}, \\
& \mathbf{v}_{\beta' 1}: \textrm{an eigenvector of }\mathbf{C}_2' \mathbf{A}_2'^{-1} \mathbf{A}_1' \mathbf{C}_1'^{-1}.\\
\end{split}
\end{align}
Note that eigenvectors may not exist for the finite Galois field. However, the existence is guaranteed by carefully choosing the encoding matrices. We provide an explicit coding scheme in the following lemma.

\begin{lemma}[$(5,3)$ E-MSR Codes]
\label{lemma-(53)}
Let $\alpha, \beta \in {\sf GF}(3)$ and be non-zero. Suppose encoding matrices are given by
\begin{align}
\begin{split}
\mathbf{A}_1 &= \left[
  \begin{array}{ccc}
    2\alpha & 0  \\
    2\beta &  \beta  \\
   \end{array}
\right], \mathbf{B}_1 = \left[
  \begin{array}{ccc}
    \alpha&  2\alpha  \\
    0 & 2\beta  \\
   \end{array}
\right], \mathbf{C}_1 = \left[
  \begin{array}{ccc}
    2\alpha & 0  \\
    \beta & 2\beta  \\
   \end{array}
\right], \\
\mathbf{A}_2 &= \left[
  \begin{array}{ccc}
    2\alpha & 0  \\
    \beta &  2\beta  \\
   \end{array}
\right],
\mathbf{B}_2 = \left[
  \begin{array}{ccc}
    \alpha&  2\alpha  \\
    0 & \beta  \\
   \end{array}
\right],
\mathbf{C}_2 = \left[
  \begin{array}{ccc}
    \alpha &  0  \\
    2\beta & 2\beta  \\
   \end{array}
\right].
\end{split}
\end{align}
Then, the code satisfies the MDS property and achieves the MSR point (\ref{eq-MSRpoint}) under exact repair constraints of all nodes.
\end{lemma}
\begin{proof}
See Appendix \ref{appen-theorem53}.
\end{proof}
\begin{remarks}
Note that encoding matrices are lower-triangular or upper-triangular. This structure has important properties. Not only does this structure guarantee invertibility, it can in fact guarantee the existence of eigenvectors. It turns out the structure as above satisfies all of the conditions needed for the MDS property and exact repair.
\end{remarks}

\begin{figure}[t]
\begin{center}
{\epsfig{figure=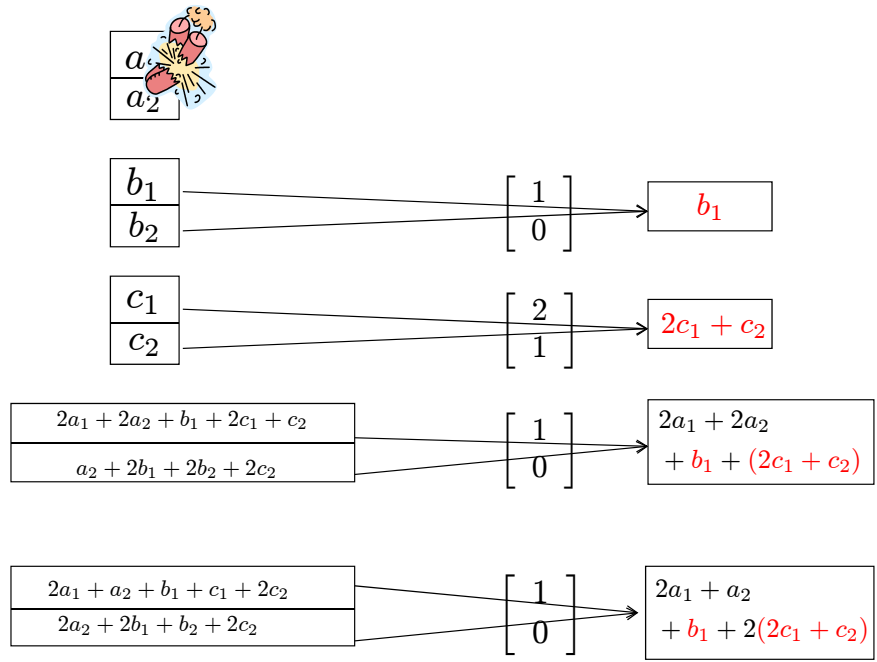, angle=0, width=0.7\textwidth}}
\end{center}
\caption{Illustration of exact repair of node 1 for a $(5,3)$ E-MSR code defined over ${\sf GF}(3)$. The eigenvector-based interference alignment scheme enables to decode 2 desired unknowns $(a_1,a_2)$ from 4 equations containing 6 unknowns. Notice that interference ``$\mathbf{b}$'' and ``$\mathbf{c}$'' are aligned simultaneously although the same projection vectors $\mathbf{v}_{\alpha 3}$ and $\mathbf{v}_{\alpha 4}$ are used.} \label{fig:53ex_GF3}
\end{figure}

\begin{example}
Fig.~\ref{fig:53ex_GF3} illustrates exact repair of node 1 $(a_1,a_2)$ for a $(5,3)$ E-MSR code defined over ${\sf GF}(3)$. Notice that interference ``$\mathbf{b}$'' and ``$\mathbf{c}$'' are aligned simultaneously. One can check exact repair of the remaining four nodes based on our proposed method.
\end{example}


\section{Conclusion}
We have systematically developed interference alignment techniques that attain the cutset-based MSR point (\ref{eq-MSRpoint}) under exact repair constraints of \emph{all} nodes. Based on the proposed framework, we provided a generalized family of codes for the cases: $(a)$ $\frac{k}{n} \leq \frac{1}{2}$; $(b)$ $k \leq 3$, for arbitrary $n \geq k$. This generalized family of codes provides insights into a dual relationship between the  systematic and parity node repair, as well as opens up a larger constructive design space of solutions. For $(5,3)$ codes which do not satisfy $\frac{k}{n} \leq \frac{1}{2}$, we have developed an eigenvector-based interference alignment to show the optimality of the cutset bound. Unlike wireless communication problems, our storage repair problems have more flexibility in designing encoding matrices which correspond to wireless channel coefficients (provided by nature) in communication problems. Exploiting this fact, we developed interference alignment techniques for optimal exact repair codes in distributed storage systems.

\appendices

\section{Proof of Lemma~\ref{lemma:TensorElementaryMatrix}}
\label{appendix:lemma_TensorEM}

It suffices to show that
\begin{align*}
\left[
   \begin{array}{ccc}
     \mathbf{A}_1'  & \mathbf{A}_2' & \mathbf{A}_3' \\
     \mathbf{B}_1'  & \mathbf{B}_2' & \mathbf{B}_3' \\
\mathbf{C}_1' & \mathbf{C}_2' & \mathbf{C}_3' \\
   \end{array}
 \right]  \left[
   \begin{array}{ccc}
     \mathbf{A}_1  & \mathbf{A}_2 & \mathbf{A}_3 \\
     \mathbf{B}_1  & \mathbf{B}_2 & \mathbf{B}_3 \\
\mathbf{C}_1 & \mathbf{C}_2 & \mathbf{C}_3 \\
   \end{array}
 \right]   = \left[
   \begin{array}{ccc}
     \mathbf{I}  & \mathbf{0} & \mathbf{0} \\
     \mathbf{0}  & \mathbf{I} & \mathbf{0} \\
 \mathbf{0}  & \mathbf{0} & \mathbf{I} \\
   \end{array}
 \right].
\end{align*}
Using (\ref{eq:63_EncodingMatrices}) and (\ref{eq-63_PrimeEM}), we compute:
\begin{align*}
& (1-\kappa^2)(\mathbf{A}_1' \mathbf{A}_1 + \mathbf{A}_2' \mathbf{B}_1  + \mathbf{A}_3' \mathbf{C}_1)  = \left( \mathbf{v}_1' \mathbf{u}_1'^t  - \kappa^2 \alpha_1' \mathbf{I} \right)   ( \mathbf{u}_1 \mathbf{v}_1^t + \alpha_1 \mathbf{I} )  \\
&\qquad +   \left( \mathbf{v}_2' \mathbf{u}_1'^t  - \kappa^2 \beta_1' \mathbf{I} \right) ( \mathbf{u}_1 \mathbf{v}_2^t + \beta_1 \mathbf{I} )  + \left( \mathbf{v}_3' \mathbf{u}_1'^t  - \kappa^2 \gamma_1' \mathbf{I} \right) ( \mathbf{u}_1 \mathbf{v}_3^t + \gamma_1 \mathbf{I} )    \\
& \quad \overset{(a)}= (\mathbf{v}_1' \mathbf{v}_1^t + \mathbf{v}_2' \mathbf{v}_2^t + \mathbf{v}_3' \mathbf{v}_3^t)   + (\alpha_1 \mathbf{v}_1' + \beta_1 \mathbf{v}_2' + \gamma_1 \mathbf{v}_3') \mathbf{u}_1'^t - \kappa^2 \mathbf{u}_1 ( \alpha_1' \mathbf{v}_1 + \beta_1' \mathbf{v}_2 + \gamma_1' \mathbf{v}_3 )^t - \kappa^2 \mathbf{I} \\
& \quad \overset{(b)}= (\mathbf{v}_1' \mathbf{v}_1^t + \mathbf{v}_2' \mathbf{v}_2^t + \mathbf{v}_3' \mathbf{v}_3^t)   + \kappa \mathbf{u}_1 \mathbf{u}_1'^t - \kappa^2 \mathbf{u}_1 ( \alpha_1' \mathbf{v}_1 + \beta_1' \mathbf{v}_2 + \gamma_1' \mathbf{v}_3 )^t - \kappa^2 \mathbf{I} \\
& \quad \overset{(c)}= (\mathbf{v}_1' \mathbf{v}_1^t + \mathbf{v}_2' \mathbf{v}_2^t + \mathbf{v}_3' \mathbf{v}_3^t)   - \kappa^2 \mathbf{I} \\
& \quad \overset{(d)}=  (1- \kappa^2) \mathbf{I} \\
\end{align*}
where $(a)$ follows from $\alpha_1 \alpha_1' + \beta_1 \beta_1' + \gamma_1 \gamma_1' = 1$ due to (\ref{eq-63-MinvM}); $(b)$ follows from (\ref{eq:uiconditions}); $(c)$ follows from
$\mathbf{u}_1' = 2 ( \alpha_1' \mathbf{v}_1 + \beta_1' \mathbf{v}_2 + \gamma_1' \mathbf{v}_3 )$ (See Claim~\ref{claim:u-and-v}); and $(d)$ follows from the fact that $\mathbf{v}_1' \mathbf{v}_1^t + \mathbf{v}_2' \mathbf{v}_2^t + \mathbf{v}_3' \mathbf{v}_3^t = \mathbf{I}$, since $(\mathbf{v}_1', \mathbf{v}_2', \mathbf{v}_3')$ are dual basis vectors.

Similarly, one can check that $\mathbf{B}_1' \mathbf{A}_2 + \mathbf{B}_2' \mathbf{B}_2  + \mathbf{B}_3' \mathbf{C}_2 = \mathbf{I}$ and $\mathbf{C}_1' \mathbf{A}_3 + \mathbf{C}_2' \mathbf{B}_3  + \mathbf{C}_3' \mathbf{C}_3 = \mathbf{I}$. Now let us compute one of the cross terms:
\begin{align*}
&(1-\kappa^2) \mathbf{A}_1' \mathbf{A}_2 + \mathbf{A}_2' \mathbf{B}_2  + \mathbf{A}_3' \mathbf{C}_2  = \left( \mathbf{v}_1' \mathbf{u}_1'^t  -\kappa^2 \alpha_1'  \mathbf{I} \right)   ( \mathbf{u}_2 \mathbf{v}_1^t + \alpha_2 \mathbf{I} )  \\
&\qquad +   \left( \mathbf{v}_2' \mathbf{u}_1'^t  -\kappa^2 \beta_1' \mathbf{I} \right) ( \mathbf{u}_2 \mathbf{v}_2^t + \beta_2 \mathbf{I} )  + \left( \mathbf{v}_3' \mathbf{u}_1'^t  - \kappa^2 \gamma_1' \mathbf{I} \right) ( \mathbf{u}_2 \mathbf{v}_3^t + \gamma_2 \mathbf{I} )    \\
& \quad \overset{(a)}=    (\alpha_2 \mathbf{v}_1' + \beta_2 \mathbf{v}_2' + \gamma_2 \mathbf{v}_3') \mathbf{u}_1'^t - \kappa^2 \mathbf{u}_2 ( \alpha_1' \mathbf{v}_1 + \beta_1' \mathbf{v}_2 + \gamma_1' \mathbf{v}_3 )^t \\
& \quad \overset{(b)}=  0 \\
\end{align*}
where $(a)$ follows from $\mathbf{u}_i'^t \mathbf{u}_j = \delta (i-j)$ and $<(\alpha_1', \beta_1', \gamma_1'), (\alpha_2, \beta_2, \gamma_2)> = 0$; $(b)$ follows from (\ref{eq:uiconditions}) and Claim~\ref{claim:u-and-v}. Similarly, we can check that the other cross terms are zero matrices. This completes the proof.

\begin{claim}
\label{claim:u-and-v}
For all $i$, $\mathbf{u}_i' = \kappa ( \alpha_i' \mathbf{v}_1 + \beta_i' \mathbf{v}_2 + \gamma_i' \mathbf{v}_3 ) $.
\end{claim}
\begin{proof}
By (\ref{eq:uiconditions}), we can rewrite
\begin{align*}
[\mathbf{u}_1, \mathbf{u}_2, \mathbf{u}_3]  = \frac{1}{\kappa} [\mathbf{v}_1', \mathbf{v}_2', \mathbf{v}_3'] \left[ \begin{array}{ccc}
    \alpha_{1} & \alpha_{2} & \alpha_{3}  \\
     \beta_{1} & \beta_{2} &  \beta_{3} \\
    \gamma_{1} & \gamma_{2} & \gamma_{3} \\
  \end{array} \right].
\end{align*}
Using the fact that  $(\mathbf{u}_1', \mathbf{u}_2', \mathbf{u}_3')$ are dual basis vectors, we get
\begin{align*}
\left[ \begin{array}{c}
    \mathbf{u}_1'^t  \\
    \mathbf{u}_2'^t \\
    \mathbf{u}_3'^t \\
  \end{array} \right]
  =  \kappa \left[
  \begin{array}{ccc}
    \alpha_{1}' & \beta_{1}' & \gamma_{1}'  \\
     \alpha_{2}' &  \beta_{2}' &  \gamma_{2}' \\
    \alpha_{3}' & \beta_{3}'   & \gamma_{3}' \\
  \end{array}
\right]  \left[ \begin{array}{c}
    \mathbf{v}_1^t  \\
    \mathbf{v}_2^t \\
    \mathbf{v}_3^t \\
  \end{array} \right].
\end{align*}
This completes the proof.
\end{proof}

\section{Proof of Theorem \ref{theorem-2kk}}
\label{appen-Theorem2kk}

For generalization, we are forced to use some heavy notation but only for this section and the related appendices. Let $\mathbf{w}_j$ be a $k$-dimensional message vector for information unit $j$. Let $\mathbf{w}_j'$ be the newly mapped information unit after remapping. Let $\mathbf{G}_{j}^{(i)}$ be an encoding matrix for parity node $i$, associated with the $j$th information unit. Let $\mathbf{G}_{j}'^{(i)}$ be the newly mapped entity.

\subsection{Exact Repair of Systematic Nodes}
For exact repair of systematic node $i$, we have each survivor node project their data
\begin{align*}
\mathbf{G}_l^{(i)} \mathbf{v}_i' &= m_l^{(i)} \mathbf{v}_i' + \mathbf{u}_l, \\
 \mathbf{G}_l^{(j)} \mathbf{v}_i' &= m_l^{(j)} \mathbf{v}_i'.
\end{align*}
Therefore, we can achieve simultaneous interference alignment for non-intended signals, while guaranteeing the decodability of desired signals.

\subsection{Exact Repair of Parity Nodes}
The idea is the same as that of Theorem~\ref{theorem-63}. The detailed procedures are as follow.
First we remap parity nodes into new variables:
\begin{align*}
\left[
   \begin{array}{c}
     \mathbf{w}_1'\\
     \mathbf{w}_2'\\
      \vdots \\
  \mathbf{w}_k'\\
   \end{array}
 \right]:=
\left[
   \begin{array}{cccc}
     \mathbf{G}_1^{(1)t}  & \mathbf{G}_2^{(1)t} & \cdots & \mathbf{G}_k^{(1)t} \\
     \mathbf{G}_1^{(2)t}  & \mathbf{G}_2^{(2)t} & \cdots & \mathbf{G}_k^{(2)t} \\
      \vdots             & \vdots              & \ddots & \vdots \\
     \mathbf{G}_1^{(k)t}  & \mathbf{G}_2^{(k)t} & \cdots & \mathbf{G}_k^{(k)t} \\
   \end{array}
 \right] \left[
   \begin{array}{c}
     \mathbf{w}_1\\
     \mathbf{w}_2\\
      \vdots \\
  \mathbf{w}_k \\
   \end{array}
 \right].
\end{align*}
Define the newly remapped encoding matrices as:
\begin{align}
\left[
   \begin{array}{cccc}
     \mathbf{G}_1'^{(1)}  & \mathbf{G}_1'^{(2)} & \cdots & \mathbf{G}_1'^{(k)} \\
     \mathbf{G}_2'^{(1)}  & \mathbf{G}_2'^{(2)} & \cdots & \mathbf{G}_2'^{(k)} \\
      \vdots              & \vdots              & \ddots & \vdots \\
     \mathbf{G}_k'^{(1)}  & \mathbf{G}_k'^{(2)} & \cdots & \mathbf{G}_k'^{(k)} \\
   \end{array}
 \right]: = \left[
   \begin{array}{cccc}
     \mathbf{G}_1^{(1)}  & \mathbf{G}_1^{(2)} & \cdots & \mathbf{G}_1^{(k)} \\
     \mathbf{G}_2^{(1)}  & \mathbf{G}_2^{(2)} & \cdots & \mathbf{G}_2^{(k)} \\
      \vdots             & \vdots              & \ddots & \vdots \\
     \mathbf{G}_k^{(1)}  & \mathbf{G}_k^{(2)} & \cdots & \mathbf{G}_k^{(k)} \\
   \end{array}
 \right]^{-1}.
\end{align}
We can now apply the generalization of  Lemma~\ref{lemma:TensorElementaryMatrix} to obtain the dual structure:

\begin{align*}
\begin{split}
\mathbf{G}_{1}'^{(1)} &= \frac{1}{1- \kappa^2} \left( \mathbf{v}_{1}' \mathbf{u}_1'^t -\kappa^2 m_1'^{(1)} \mathbf{I}\right),
\cdots,
\mathbf{G}_{k}'^{(1)} = \frac{1}{1-\kappa^2} \left( \mathbf{v}_{1}' \mathbf{u}_k'^t -\kappa^2 m_1'^{(k)} \mathbf{I}\right), \\
&\vdots \qquad \qquad \qquad  \ddots \qquad \qquad \qquad  \vdots \\
\mathbf{G}_{1}'^{(k)} &= \frac{1}{1-\kappa^2} \left(\mathbf{v}_{k}' \mathbf{u}_1'^t -\kappa^2 m_k'^{(1)} \mathbf{I} \right),
\cdots,
\mathbf{G}_{k}'^{(k)} = \frac{1}{1-\kappa^2} \left( \mathbf{v}_{k}' \mathbf{u}_k'^t -\kappa^2 m_k'^{(k)} \mathbf{I} \right), \\
\end{split}
\end{align*}
where the dual basis vectors are defined as:
\begin{align*}
\left[
                \begin{array}{cccc}
                  m_1'^{(1)} & m_2'^{(1)} & \cdots    & m_{k}'^{(1)} \\
                  m_1'^{(2)} & m_2'^{(2)} & \cdots    & m_{k}'^{(2)} \\
                  \vdots    & \vdots    & \ddots    & \vdots \\
                  m_1'^{(k)} & m_2'^{(k)} & \cdots    & m_{k}'^{(k)} \\
                \end{array}
              \right]: = \left[
                \begin{array}{cccc}
                  m_1^{(1)} & m_1^{(2)} & \cdots    & m_{1}^{(k)} \\
                  m_2^{(1)} & m_2^{(2)} & \cdots    & m_{2}^{(k)} \\
                  \vdots    & \vdots    & \ddots    & \vdots \\
                  m_k^{(1)} & m_k^{(2)} & \cdots    & m_{k}^{(k)} \\
                \end{array}
              \right]^{-1}.
\end{align*}

Let us check exact repair of parity node $i$. We choose projection vectors as $\mathbf{u}_i$.
Then, $\forall l=1,\cdots, k$, we get:
\begin{align*}
(1-\kappa^2)\mathbf{G}_l'^{(i)} \mathbf{u}_i &= -\kappa^2 m_i^{(l)} \mathbf{u}_i + \mathbf{v}_l', \\
(1-\kappa^2) \mathbf{G}_l'^{(j)} \mathbf{u}_i &= -\kappa^2 m_j^{(l)} \mathbf{u}_i.
\end{align*}
Therefore, we can achieve simultaneous interference alignment for non-intended signals, while guaranteeing the decodability of desired signals.

\subsection{The MDS-Code Property}
We check the invertibility of a composite encoding matrix when a Data Collector connects to $i$ systematic nodes and $(k-i)$ parity nodes for $i=0,\cdots,k$. The main idea is to use a Gaussian elimination method as we did in Section~\ref{sec-MDScodeProperty}. The verification is tedious and therefore  details are omitted.

\subsection{Minimum Required Finite-Field Size}
Note that the dimension of an encoding matrix is $k$-by-$k$. Therefore, the minimum finite-field size required to generate a Cauchy matrix is $2k$, i.e., $q \geq 2k$.

\section{Proof of Lemma \ref{lemma-(53)}}
\label{appen-theorem53}

\subsection{Exact Repair}
With the Gaussian elimination method, we get
\begin{align}
\begin{split}
\mathbf{A}_1' &= \left[
  \begin{array}{cc}
    \frac{1}{\alpha} &  \frac{1}{\beta}  \\
    \frac{1}{\alpha} & 0 \\
  \end{array}
\right], \mathbf{B}_1' = \left[
   \begin{array}{cc}
    \frac{1}{\alpha} &  \frac{2}{\beta}  \\
    \frac{1}{\alpha} & 0 \\
  \end{array}
\right], \mathbf{C}_1' = \left[
  \begin{array}{cc}
    0  & \frac{2 \alpha}{\beta} \\
   \frac{2 \beta}{ \alpha }    &1 \\
  \end{array}
\right], \\
\mathbf{A}_2' &= \left[
  \begin{array}{cc}
     0 & \frac{1}{\beta} \\
 \frac{1}{\alpha}&  \frac{1}{\beta}  \\
  \end{array}
\right],
\mathbf{B}_2' = \left[
  \begin{array}{cc}
     0 &  \frac{2}{\beta} \\
 \frac{1}{\alpha} &  \frac{2}{\beta} \\
  \end{array}
\right],
\mathbf{C}_2' = \left[
  \begin{array}{cc}
     0 &  \frac{2 \alpha}{\beta}  \\
  \frac{2 \beta}{ \alpha }   &1 \\
  \end{array}
\right].
\end{split}
\end{align}
Using this, we can easily check the the existence of eigenvectors (\ref{eq-53conditions3}) and decodabiity of desired signals (\ref{eq-53conditions}). This completes the proof.

\subsection{The MDS-Code Property}
Obviously, all the encoding matrices are invertible due to their lower-triangular or upper-triangular structure. We consider three cases where a Data Collector connects to (1) 3 systematic nodes; (2) 2 systematic nodes and 1 parity node; and (3) 1 systematic node and 2 parity nodes. The first is a trivial case where the composite matrix associated with information units is an identity matrix.
The second case is also trivial, since each encoding matrix is invertible so that the composite matrix is invertible as well. For the last case, we consider
\begin{align}
 \left[
   \begin{array}{ccc}
     \mathbf{A}_1 & \mathbf{A}_2 & \mathbf{0} \\
      \mathbf{B}_1 & \mathbf{B}_2 & \mathbf{0}\\
      \mathbf{C}_1 & \mathbf{C}_2 & \mathbf{I} \\
   \end{array}
 \right]= \left[
  \begin{array}{cc|cc|cc}
    2\alpha &  0&  2\alpha     &0    &0          &0 \\
  2\beta &  \beta &  \beta     &2\beta    &0          &0 \\
\hline
  \alpha & 2\alpha  &  \alpha     &2 \alpha    &0          &0 \\
  0 &  2\beta &  0     &\beta    &0          &0 \\
\hline
2\alpha &  0 &  \alpha     &0    &1          &0 \\
  \beta &  2\beta &  2\beta    &2\beta    &0          &1 \\
  \end{array}
\right].
\end{align}
It is easy to check the invertibility of this matrix via the Gaussian elimination method.
The invertibility for all the cases guarantees the MDS property.

\section*{Acknowledgment}
We gratefully acknowledge Prof. P. V. Kumar (of IISc) and his students, N. B. Shah and  K. V. Rashmi, for insightful discussions and fruitful collaboration related to the structure of exact regeneration codes.

\bibliographystyle{IEEEtran}
\bibliography{Storage_IA}

\begin{thebibliography}{10}
\providecommand{\url}[1]{#1}
\csname url@samestyle\endcsname
\providecommand{\newblock}{\relax}
\providecommand{\bibinfo}[2]{#2}
\providecommand{\BIBentrySTDinterwordspacing}{\spaceskip=0pt\relax}
\providecommand{\BIBentryALTinterwordstretchfactor}{4}
\providecommand{\BIBentryALTinterwordspacing}{\spaceskip=\fontdimen2\font plus
\BIBentryALTinterwordstretchfactor\fontdimen3\font minus
  \fontdimen4\font\relax}
\providecommand{\BIBforeignlanguage}[2]{{%
\expandafter\ifx\csname l@#1\endcsname\relax
\typeout{** WARNING: IEEEtran.bst: No hyphenation pattern has been}%
\typeout{** loaded for the language `#1'. Using the pattern for}%
\typeout{** the default language instead.}%
\else
\language=\csname l@#1\endcsname
\fi
#2}}
\providecommand{\BIBdecl}{\relax}
\BIBdecl

\bibitem{Dimakis:INFOCOM}
A.~G. Dimakis, P.~B. Godfrey, Y.~Wu, M.~Wainwright, and K.~Ramchandran,
  ``Network coding for distributed storage systems,'' \emph{IEEE INFOCOM},
  2007.

\bibitem{Wu:Allerton}
Y.~Wu, A.~G. Dimakis, and K.~Ramchandran, ``Deterministic regenerating codes
  for distributed storage,'' \emph{Allerton Conference on Control, Computing
  and Communication}, Sep. 2007.

\bibitem{Sameer:ISIT2010}
S.~Pawar, S.~E. Rouayheb, and K.~Ramchandran, ``On secure distributed data
  storage under repair dynamics,'' \emph{submitted to IEEE ISIT}, 2010.

\bibitem{KumarRamchandran_MSR}
N.~B. Shah, K.~V. Rashmi, P.~V. Kumar, and K.~Ramchandran, ``Explicit codes
  minimizing repair bandwidth for distributed storage,'' \emph{IEEE ITW, Jan.
  2010, online avaiable at arXiv:0908.2984v2}, Sep. 2009.

\bibitem{Cullina_MSR}
D.~Cullina, A.~G. Dimakis, and T.~Ho, ``Searching for minimum storage
  regenerating codes,'' \emph{Allerton Conference on Control, Computing and
  Communication}, Sep. 2009.

\bibitem{Mohammad}
M.~A. Maddah-Ali, S.~A. Motahari, and A.~K. Khandani, ``Communication over
  {MIMO} {X} channels: Interference alignment, decomposition, and performance
  analysis,'' \emph{IEEE Transactions on Information Theory}, vol.~54, pp.
  3457--3470, Aug. 2008.

\bibitem{Jafar:IC}
V.~R. Cadambe and S.~A. Jafar, ``Interference alignment and the degree of
  freedom for the {K} user interference channel,'' \emph{IEEE Transactions on
  Information Theory}, vol.~54, no.~8, pp. 3425--3441, Aug. 2008.

\bibitem{Suh:Allerton}
C.~Suh and D.~Tse, ``Interference alignment for cellular networks,''
  \emph{Allerton Conference on Control, Computing and Communication}, Sep.
  2008.

\bibitem{Wu:ISIT}
Y.~Wu and A.~G. Dimakis, ``Reducing repair traffic for erasure coding-based
  storage via interference alignment,'' \emph{Proc. of IEEE ISIT}, 2009.

\bibitem{KumarRamchandran_MBR}
K.~V. Rashmi, N.~B. Shah, P.~V. Kumar, and K.~Ramchandran, ``Explicit
  construction of optimal exact regenerating codes for distributed storage,''
  \emph{Allerton Conference on Control, Computing and Communication}, Sep.
  2009.

\bibitem{Wu:PartialExact}
Y.~Wu, ``A construction of systematic mds codes with minimum repair
  bandwidth,'' \emph{arXiv:0910.2486}, Oct. 2009.

\bibitem{Koetter:it}
R.~Koetter and M.~Medard, ``An algebraic approach to network coding,''
  \emph{IEEE/ACM Transactions on Networking}, vol.~11, no.~5, Oct. 2003.

\bibitem{ahlswede:it}
R.~Ahlswede, N.~Cai, S.-Y.~R. Li, and R.~W. Yeung, ``Network information
  flow,'' \emph{IEEE Transactions on Information Theory}, vol.~46, no.~4, pp.
  1204--1216, Jul. 2000.

\bibitem{Householder}
A.~S. Householder, \emph{The Theory of Matrices in Numerical Analysis}.\hskip
  1em plus 0.5em minus 0.4em\relax Dover, Toronto, Cananda, 1974.

\bibitem{Dubrulle}
A.~A. Bubrulle, ``Work notes on elementary matrices,'' \emph{Tech. Rep.
  HPL-93-69, Hewlett-Packard Laboratory}, 1993.

\bibitem{Bernstein:Cauchy}
D.~S. Bernstein, \emph{Matrix mathematics: Theory, facts, and formulas with
  application to linear systems theory}.\hskip 1em plus 0.5em minus 0.4em\relax
  Princeton University Press, 2005.

\end{thebibliography}

\end{document}